\documentclass{lmcs}
\pdfoutput=1

\usepackage{lastpage}
\lmcsdoi{19}{1}{7}
\lmcsheading{}{\pageref{LastPage}}{}{}%
{Dec.~01,~2021}{Jan.~18,~2023}{}

\usepackage[utf8]{inputenc}
\usepackage{amsmath,amssymb,stmaryrd,mathpartir}
\usepackage{hyperref}
\usepackage{graphicx,graphbox}
\usepackage{cite}

\DeclareSymbolFont{symbolsC}{U}{txsyc}{m}{n}
\SetSymbolFont{symbolsC}{bold}{U}{txsyc}{bx}{n}
\DeclareFontSubstitution{U}{txsyc}{m}{n}
\DeclareMathSymbol{\multimapinv}{\mathrel}{symbolsC}{18}
\DeclareMathSymbol{\multimapdot}{\mathrel}{symbolsC}{20}
\DeclareMathSymbol{\multimapdotinv}{\mathrel}{symbolsC}{21}
\DeclareMathSymbol{\multimap}{\mathrel}{AMSa}{40}

\newtheorem{observation}[thm]{Observation}

\newcommand{\A}{\mathbb{A}}
\newcommand{\B}{\mathfrak{B}}
\newcommand{\D}{\mathbb{D}}
\newcommand{\X}{\mathbb{X}}
\newcommand{\bh}{\mathbf{H}}
\newcommand{\bv}{\mathbf{V}}

\newcommand{\op}{\mathsf{op}}
\newcommand{\ex}[1]{{#1}^{\circ\bullet}}

\newcommand{\corner}[1]{{\text{}^\ulcorner_\llcorner\!{#1}\!_\lrcorner^\urcorner}}

\newcommand{\alice}{\texttt{Alice}}
\newcommand{\bob}{\texttt{Bob}}

\makeatletter
\providecommand{\leftsquigarrow}{%
  \mathrel{\mathpalette\reflect@squig\relax}%
}
\newcommand{\reflect@squig}[2]{%
  \reflectbox{$\m@th#1\rightsquigarrow$}%
}
\makeatother

\begin{document}

\title{Concurrent Process Histories and Resource Transducers}
\author[C.~Nester]{Chad Nester}
\address{Tallinn University of Technology}
\thanks{This research was supported by the ESF funded Estonian IT Academy research measure (project 2014-2020.4.05.19-0001). This paper is a revised and extended version of~\cite{Nester2021Concurrent}.}

\begin{abstract}
We identify the algebraic structure of the material histories generated by concurrent processes. Specifically, we extend existing categorical theories of resource convertibility to capture concurrent interaction. Our formalism admits an intuitive graphical presentation via string diagrams for proarrow equipments. We also consider certain induced categories of resource transducers, which are of independent interest due to their unusual structure.
\end{abstract}

\maketitle

\section{Introduction}

Concurrent systems are abundant in computing, and indeed in the world at large. Despite the large amount of attention paid to the modelling of concurrency in recent decades (e.g.,~\cite{Hoa78,Mil80, Pet96, Mil99, Abr14}), a canonical mathematical account has yet to emerge, and the basic structure of concurrent systems remains elusive.

In this paper we present a basic structure that captures what we will call the \emph{material} aspect of concurrent systems: As a process unfolds in time it leaves behind a material history of effects on the world, like the way a slug moving through space leaves a trail of slime. This slime is captured in a natural way by \emph{resource theories} in the sense of~\cite{Coe14}, in which morphisms of symmetric monoidal categories --- conveniently expressed as string diagrams --- are understood as transformations of resources.

\[
\includegraphics[height=3cm,align=c]{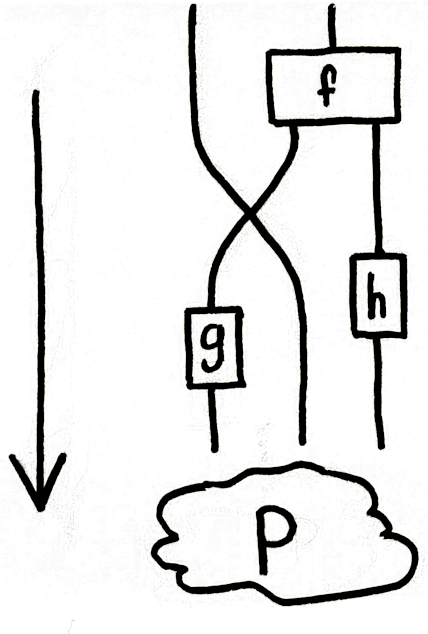}
\hspace{1cm}
\leftrightsquigarrow
\hspace{1cm}
\includegraphics[height=3cm,align=c]{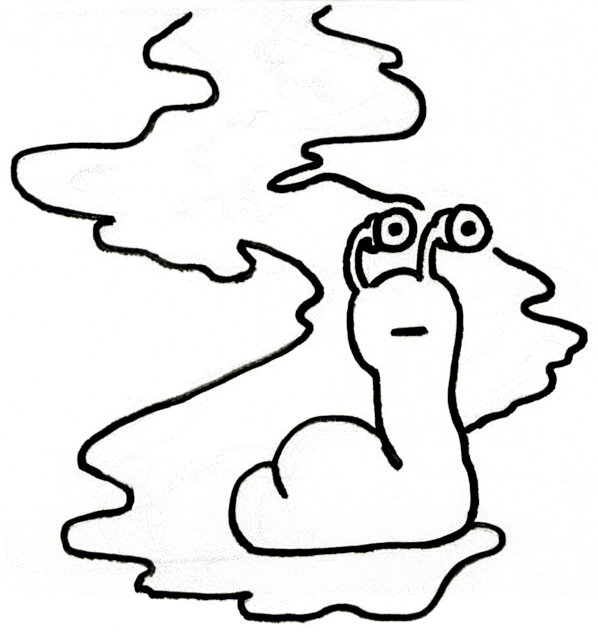}
\]

From the resource theoretic perspective, objects of a symmetric monoidal category are understood as collections of resources, with the unit object denoting the empty collection and the tensor product of two collections consisting of their combined contents. Morphisms are understood as ways to transform one collection of resources into another, which may be combined sequentially via composition, and in parallel via the tensor product. For example, the process of baking bread might generate the following material history:
\[
\includegraphics[height=3cm,align=c]{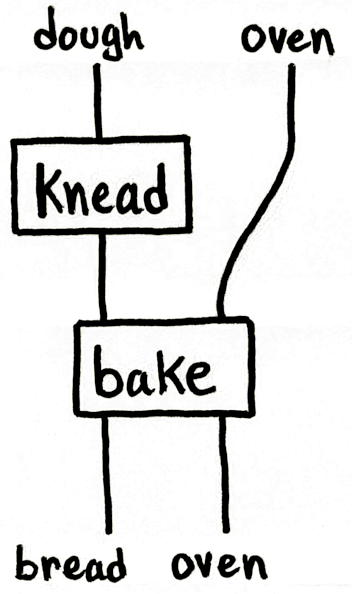}
\]
meaning that the baking process involved kneading dough and baking it in an oven to obtain bread (and also the oven).

This approach to expressing the material history of a process has many advantages: It is general, in that it assumes minimal structure; canonical, in that monoidal categories are well-studied as mathematical objects; and relatively friendly, as it admits an intuitive graphical calculus (string diagrams). However, it is unable to capture the interaction between components of a concurrent process. For example, consider our hypothetical baking process and suppose that the kneading and baking of the dough are handled by separate subsystems, with control of the dough being handed to the baking subsystem once the kneading is complete. Such interaction of parts is a fundamental aspect of concurrency, but is not expressible in this framework --- we can only describe the effects of the system as a whole.

We remedy this by extending a given resource theory to allow the decomposition of material histories into concurrent components. Specifically, we augment the string diagrams for symmetric monoidal categories with \emph{corners}, through which resources may flow between different components of a transformation.

\[
\includegraphics[height=3cm,align=c]{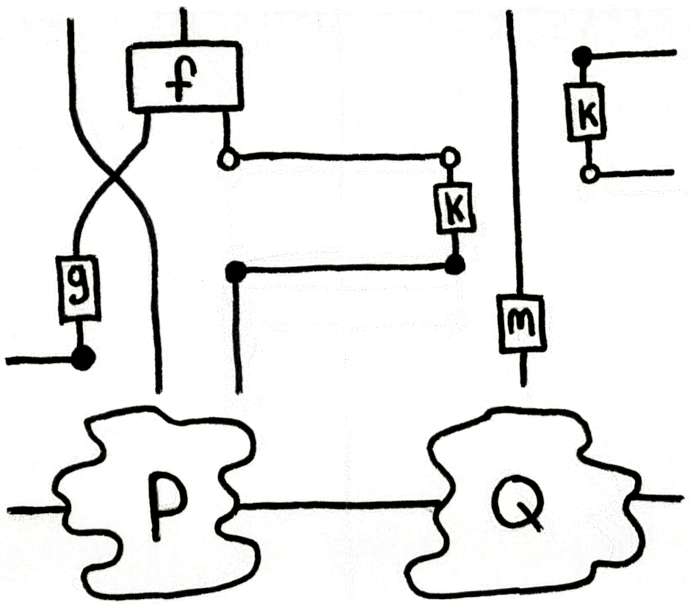}
\hspace{1cm}
\leftrightsquigarrow
\hspace{1cm}
\includegraphics[height=3cm,align=c]{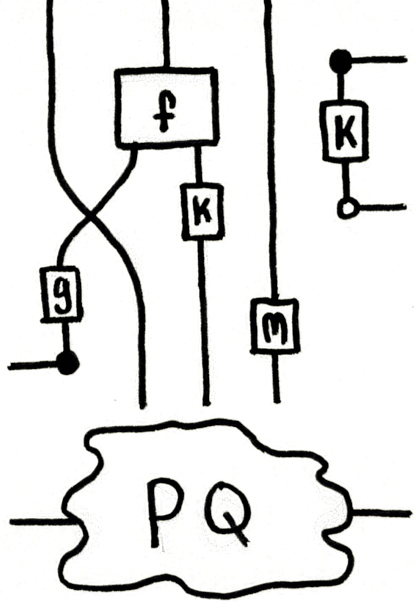}
\]
Returning to our baking example, we might express the material history of the kneading and baking subsystems \emph{separately} with the following diagrams, which may be composed horizontally to obtain the material history of the baking process as a whole.
\[
\includegraphics[height=1.7cm,align=c]{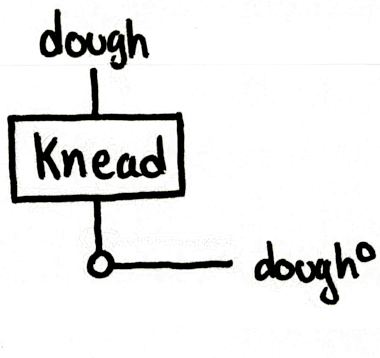}
\hspace{1cm}
\includegraphics[height=1.7cm,align=c]{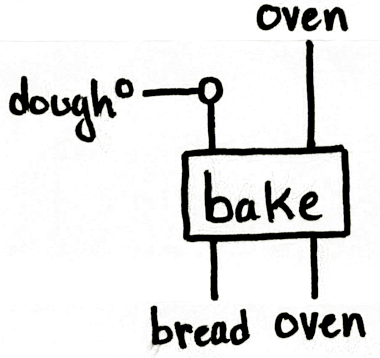}
\]

These augmented diagrams denote cells of a single-object double category constructed from the original resource theory. The corners make this double category into a proarrow equipment, which turns out to be all the additional structure we need in order to express concurrent interaction. From only this structure, we obtain a theory of exchanges --- a sort of minimal system of behavioural types --- that conforms to our intuition about how such things ought to work remarkably well.

Our approach to these concurrent material histories retains the aforementioned advantages of the resource-theoretic perspective: We lose no generality, since our construction applies to any resource theory; It is canonical, with proarrow equipments being a fundamental structure in formal category theory --- although not usually seen in such concrete circumstances; Finally, it remains relatively friendly, since the string diagrams for monoidal categories extend in a natural way to string diagrams for proarrow equipments~\cite{Mye16}.

Every single-object double category defines two monoidal categories: one composed of cells with trivial left and right boundary, and one composed of cells with trivial top and bottom boundary. For the double category obtained by adding corners to a resource theory the induced monoidal categories are, respectively, the resource theory itself and a category of \emph{resource transducers} --- being an alternative interpretation of concurrent transformations that neither begin nor end with any resources. This category of resource transducers is rich in structure, exhibiting unusual features that make it an interesting object of study in its own right. We establish some elementary properties of this category and axiomatize it directly --- that is, we give a monoidal signature and a collection of equations that characterize the category of resource transducers.

This paper is an extended version of~\cite{Nester2021Concurrent}, including additional examples and an exploration of the aforementioned categories of resource transducers.

\subsection{Contributions and Related Work}

\emph{Related Work}. Monoidal categories are ubiquitous --- if often implicit --- in theoretical computer science. An example from the theory of concurrency is~\cite{Mes90}, in which monoidal categories serve a purpose similar to their purpose here. String diagrams for monoidal categories seem to have been invented independently a number of times, but until recently were uncommon in printed material due to technical limitations. The usual reference is~\cite{Joy91}. We credit the resource-theoretic interpretation of monoidal categories and their string diagrams to~\cite{Coe14}. Double categories first appear in~\cite{Ehr63}. Free double categories are considered in~\cite{Daw02} and again in~\cite{Fio08}. The idea of a proarrow equipment first appears in~\cite{Woo82}, albeit in a rather different form. Proarrow equipments have subsequently appeared under many names in formal category theory (see e.g.,~\cite{Shu08,Gra04}). String diagrams for double categories and proarrow equipments are treated precisely in~\cite{Mye16}. We have been inspired by work on message passing and behavioural types, in particular~\cite{Coc09}, from which we have adopted our notation for exchanges.

\noindent \emph{Contributions}. The main contribution of this paper is the resource-theoretic interpretation of the free cornering and the observation that it captures the structure of concurrent process histories. Other contributions concern the categorical structure of the free cornering of a resource theory: we show that it has \emph{crossing cells} and is consequently a monoidal double category in Lemma~\ref{lem:technicalcrossing} and Lemma~\ref{lem:monoidaldouble}, argue that the vertical cells are the original monoidal category in Proposition~\ref{prop:verticaloriginal}, show that the induced monoidal category of horizontal cells can be understood as a category of \emph{resource transducers}, and establish Lemma~\ref{lem:spatial}, Lemma~\ref{lem:moraliso}, Observation~\ref{obs:halfsym}, Lemma~\ref{lem:strongmonoidal}, Lemma~\ref{lem:involution}, and Proposition~\ref{prop:linearactegory} --- all of which concern the structure of this category of horizontal cells. Finally, we give an axiomatization of the category of horizontal cells in terms of equations over a monoidal signature in Section~\ref{sec:axiomscheme}. The original contributions of this paper over~\cite{Nester2021Concurrent} are Lemma~\ref{lem:spatial}, Lemma~\ref{lem:strongmonoidal}, Lemma~\ref{lem:involution}, Proposition~\ref{prop:linearactegory}, and the axiom scheme of Section~\ref{sec:axiomscheme}.

\subsection{Organization and Prerequisites}

\emph{Prerequisites}. This paper is largely self-contained, but we assume some familiarity with category theory, in particular with monoidal categories and their string diagrams. Some good references are~\cite{Mac71,Sel10,Fon18}.

\noindent \emph{Organization}. In Section~\ref{sec:resourcetheories} we review the resource-theoretic interpretation of symmetric monoidal categories. We continue by reviewing the theory of double categories in Section~\ref{sec:singledouble}, specialized to the single object case. In Section~\ref{sec:corneringcrossing} we recall the notion of proarrow equipment, introduce the free cornering of a resource theory, and exhibit the existence of crossing cells in the free cornering. In Section~\ref{sec:interpretation} we show how the free cornering of a resource theory inherits its resource-theoretic interpretation while enabling the concurrent decomposition of resource transformations. In Section~\ref{sec:transducers} we consider the category of resource transducers and investigate its structure, and in Section~\ref{sec:axiomscheme} we give an axiom scheme for it. In Section~\ref{sec:conclusion} we conclude and consider directions for future work.

\section{Monoidal Categories as Resource Theories}\label{sec:resourcetheories}

Symmetric strict\footnote{We work with strict monoidal categories for the sake of convenience and readability. We expect the present development to apply equally well to the general case, and if pressed would appeal to the coherence theorem for monoidal categories~\cite{Mac71}.} monoidal categories can be understood as theories of resource transformation. Objects are interpreted as collections of resources, with $A \otimes B$ the collection consisting of both $A$ and $B$, and $I$ the empty collection. Arrows $f : A \to B$ are understood as ways to transform the resources of $A$ into those of $B$. We call symmetric strict monoidal categories \emph{resource theories} when we have this sort of interpretation in mind.

For example, let $\mathfrak{B}$ be the free symmetric strict monoidal category with generating objects
\[
\{ \texttt{bread},\texttt{dough},\texttt{water},\texttt{flour},\texttt{oven}\}
\]
and with generating arrows
\begin{mathpar}

  \texttt{mix} : \texttt{water} \otimes \texttt{flour} \to \texttt{dough}

  \texttt{knead} : \texttt{dough} \to \texttt{dough}

  \texttt{bake} : \texttt{dough} \otimes \texttt{oven} \to \texttt{bread} \otimes \texttt{oven}

\end{mathpar}
subject to no equations. $\mathfrak{B}$ can be understood as a resource theory of baking bread. The arrow $\texttt{mix}$ represents the process of combining water and flour to form a bread dough, \texttt{knead} represents kneading dough, and \texttt{bake} represents baking dough in an oven to obtain bread (and an oven).

The structure of symmetric strict monoidal categories provides natural algebraic scaffolding for composite transformations. For example, consider the following arrow of $\mathfrak{B}$:
\[
(\texttt{bake} \otimes 1_{\texttt{dough}});(1_\texttt{bread} \otimes \sigma_{\texttt{oven},\texttt{dough}};\texttt{bake})
\]
of type
\[
\texttt{dough} \otimes \texttt{oven} \otimes \texttt{dough} \to \texttt{bread} \otimes \texttt{bread} \otimes \texttt{oven}
\]
where $\sigma_{A,B} : A \otimes B \stackrel{\sim}{\to} B \otimes A$ is the braiding. This arrow describes the transformation of two units of dough into loaves of bread by baking them one after the other in an oven.

It is often more intuitive to write composite arrows like this as string diagrams: Objects are depicted as wires, and arrows as boxes with inputs and outputs. Composition is represented by connecting output wires to input wires, and we represent the tensor product of two morphisms by placing them beside one another. Finally, the braiding is represented by crossing the wires involved. For the morphism discussed above, the corresponding string diagram is:
\[
\includegraphics[height=3.5cm]{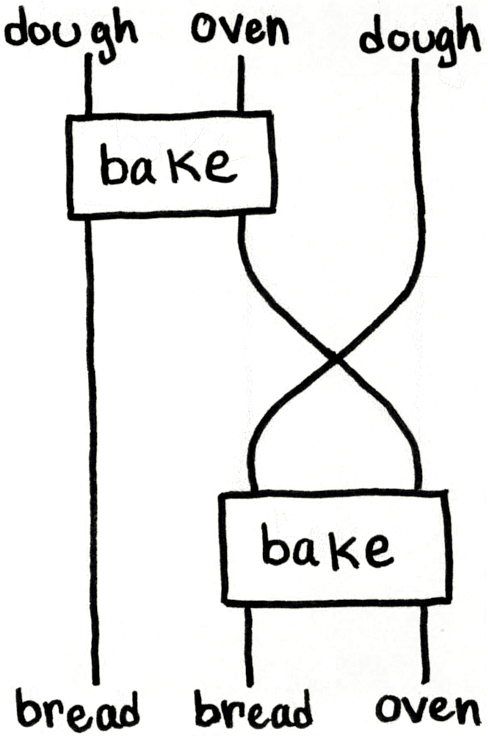}
\]
Notice how the topology of the diagram captures the logical flow of resources.

Given a pair of parallel arrows $f,g : A \to B$ in some resource theory, both $f$ and $g$ are ways to obtain $B$ from $A$, but they may not have the same effect on the resources involved. We explain by example: Consider the parallel arrows $1_{\texttt{dough}},\texttt{knead} : \texttt{dough} \to \texttt{dough}$ of $\mathfrak{B}$. Clearly these should not be understood to have the same effect on the dough in question, and this is reflected in $\mathfrak{B}$ by the fact that they are not made equal by its axioms. Similarly, $\texttt{knead}$ and $\texttt{knead} \circ \texttt{knead}$ are not equal in $\mathfrak{B}$, which we understand to mean that kneading dough twice does not have the same effect as kneading it once, and that in turn any $\texttt{bread}$ produced from twice-kneaded dough will be different from once-kneaded bread in our model.

Consider a hypothetical resource theory constructed from $\mathfrak{B}$ by imposing the equation $\texttt{knead} \circ \texttt{knead} = \texttt{knead}$. In this new setting we understand kneading dough once to have the same effect as kneading it twice, three times, and so on, because the corresponding arrows are all equal. Of course, the sequence of events described by $\texttt{knead}$ is not the one described by $\texttt{knead} \circ \texttt{knead}$: In the former the dough has been kneaded only once, while in the latter it has been kneaded twice. The equality of the two string diagrams indicates that these two different processes would have the same effect on the dough involved. We adopt as a general principle in our design and understanding of resource theories that transformations should be equal as morphisms if and only if they have the same effect on the resources involved.

For the sake of further illustration, observe that by naturality of the braiding maps the following two resource transformations are equal in $\mathfrak{B}$:
\[
\includegraphics[height=7cm]{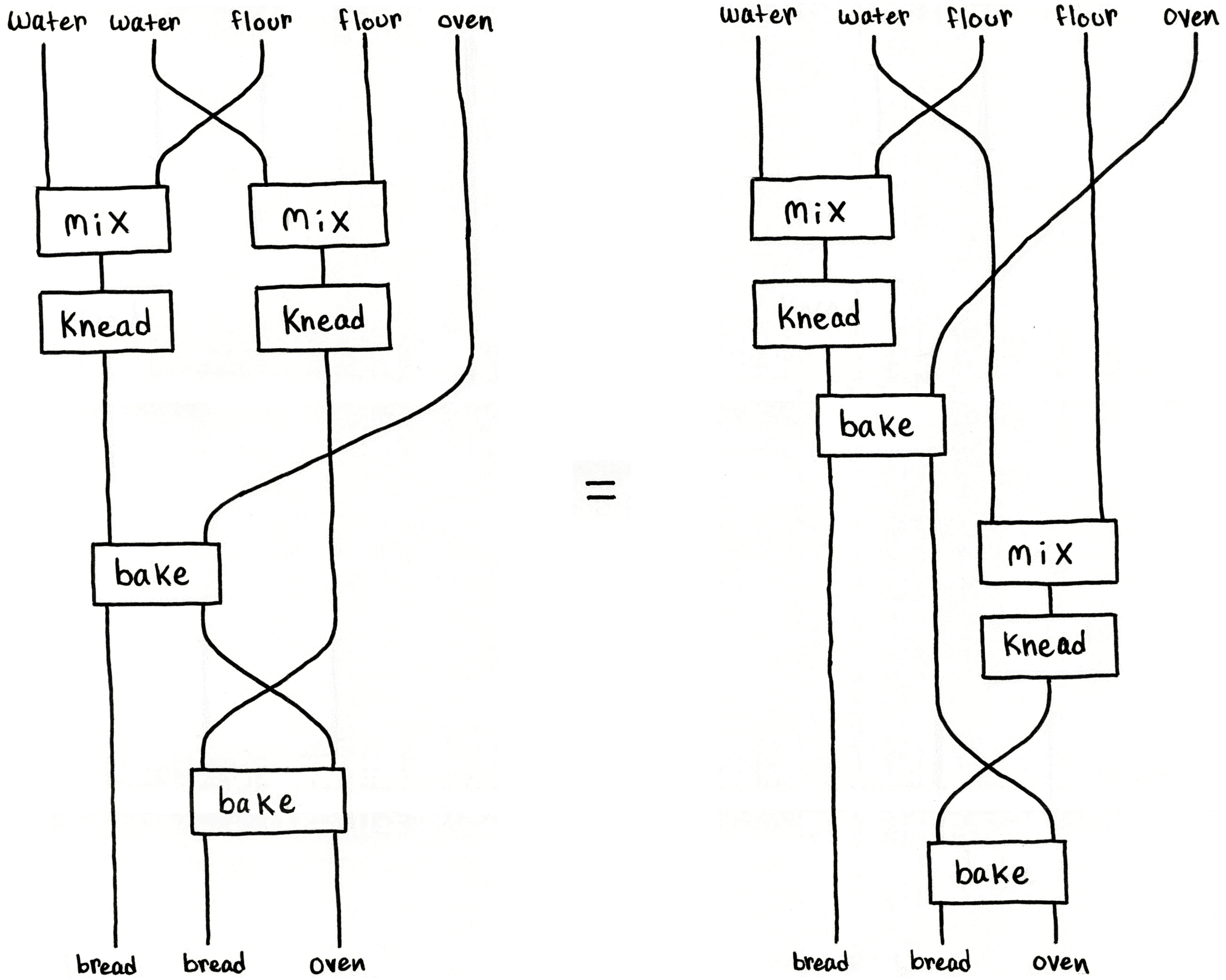}
\]
Each transformation gives a method of baking two loaves of bread. On the left, two batches of dough are mixed and kneaded before being baked one after the other. On the right, first one batch of dough is mixed, kneaded and baked and only then is the second batch mixed, kneaded, and baked. Their equality tells us that, according to $\mathfrak{B}$, the two procedures will have the same effect, resulting in the same bread when applied to the same ingredients with the same oven.

\section{Single-Object Double Categories}\label{sec:singledouble}

In this section we set up the rest of our development by presenting the theory of \emph{single-object double categories}, being those double categories $\D$ with exactly one object. In this case $\D$ consists of a \emph{horizontal edge monoid} $\D_H = (\D_H, \otimes, I)$, a \emph{vertical edge monoid} $\D_V = (\D_V, \otimes, I)$, and a collection of \emph{cells}
\[
\includegraphics[height=1.7cm,align=c]{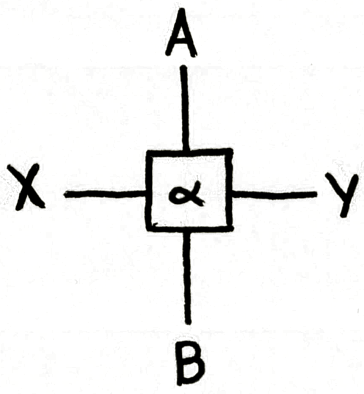}
\]
where $A,B \in \D_H$ and $X,Y \in \D_V$. Given cells $\alpha,\beta$ where the right boundary of $\alpha$ matches the left boundary of $\beta$ we may form a cell $\alpha \vert \beta$ --- their \emph{horizontal composite} --- and similarly if the bottom boundary of $\alpha$ matches the top boundary of $\beta$ we may form $\frac{\alpha}{\beta}$ --- their \emph{vertical composite} --- with the boundaries of the composite cell formed from those of the component cells using $\otimes$. We depict horizontal and vertical composition, respectively, as in:
\[
\includegraphics[height=1.7cm,align=c]{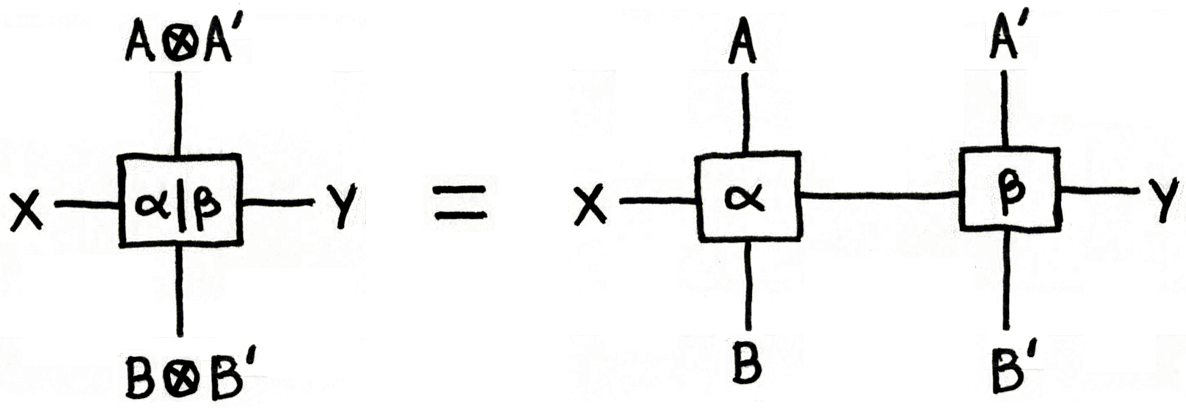}
\hspace{0.5cm}
\text{and}
\hspace{0.5cm}
\includegraphics[height=2.3cm,align=c]{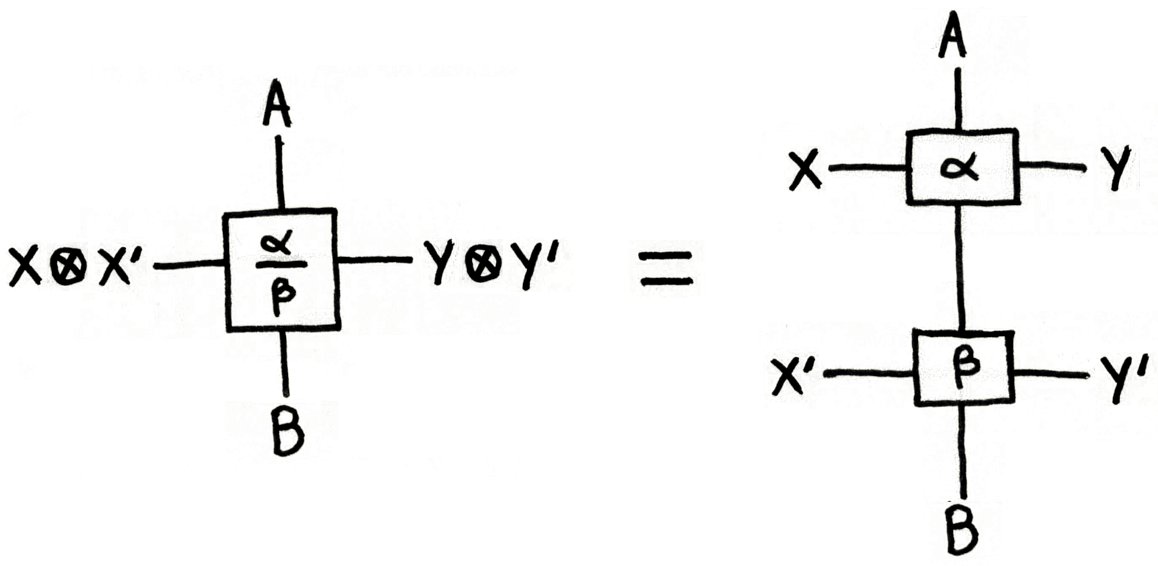}
\]
Horizontal and vertical composition of cells are required to be associative and unital. We omit wires of sort $I$ in our depictions of cells, allowing us to draw horizontal and vertical identity cells, respectively, as in:
\[
\includegraphics[height=1.7cm,align=c]{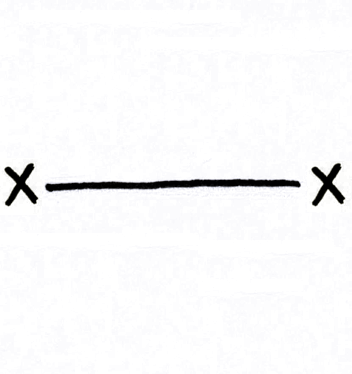}
\hspace{1cm}
\text{and}
\hspace{1cm}
\includegraphics[height=1.7cm,align=c]{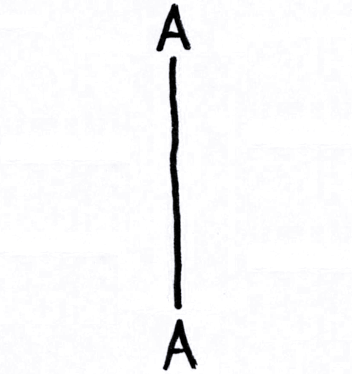}
\]
Finally, the horizontal and vertical identity cells of type $I$ must coincide --- we write this cell as $\square_I$ and depict it as empty space, see below on the left --- and vertical and horizontal composition must satisfy the interchange law. That is, $\frac{\alpha}{\beta}\vert \frac{\gamma}{\delta} = \frac{\alpha \vert \gamma}{\beta \vert \delta}$, allowing us to unambiguously interpret the diagram below on the right:
\[
\includegraphics[height=1.7cm,align=c]{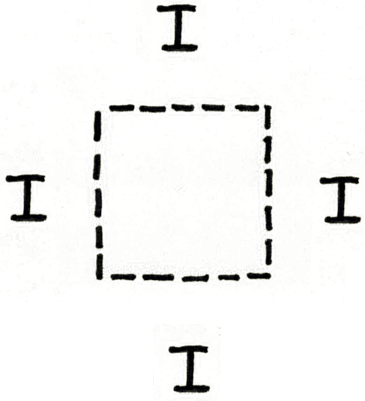}
\hspace{3cm}
\includegraphics[height=2.3cm,align=c]{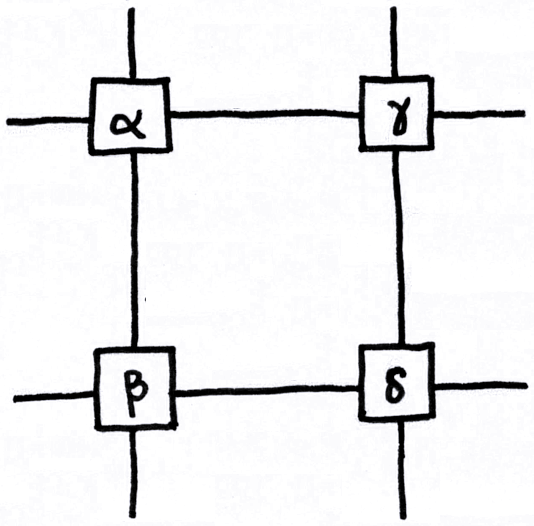}
\]

Every single-object double category $\D$ defines strict monoidal categories $\bv \D$ and $\bh \D$, consisting of the cells for which the $\D_H$ and $\D_V$ valued boundaries respectively are all $I$, as in:
\[
\includegraphics[height=1.7cm,align=c]{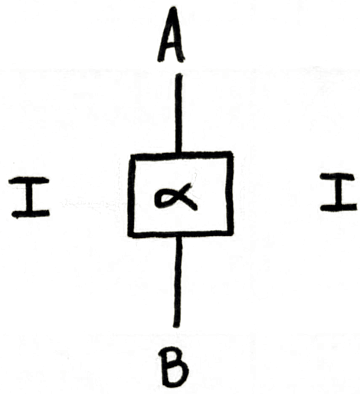}
\hspace{1cm}
\text{and}
\hspace{1cm}
\includegraphics[height=1.7cm,align=c]{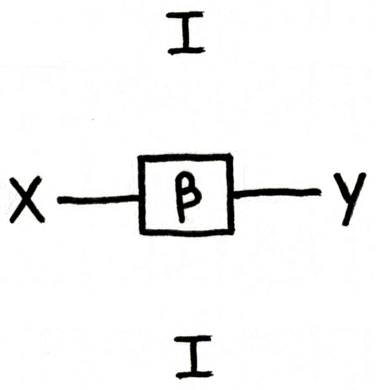}
\]
That is, the collection of objects of $\bv \D$ is $\D_H$, composition in $\bv \D$ is vertical composition of cells, and the tensor product in $\bv \D$ is given by horizontal composition:
\[
\includegraphics[height=1.7cm]{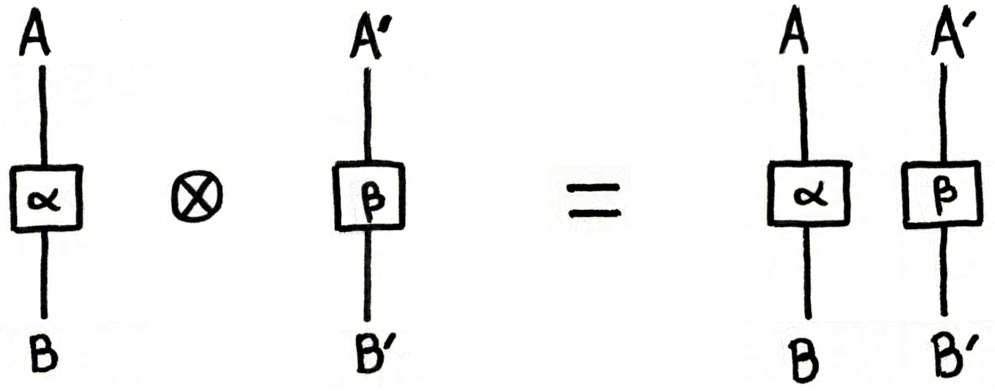}
\]
In this way, $\bv \D$ forms a strict monoidal category, which we call the category of $\emph{vertical cells}$ of $\D$. Similarly, $\bh \D$ is also a strict monoidal category (with collection of objects $\D_V$) which we call the \emph{horizontal cells} of $\D$.

\section{Cornerings and Crossings}\label{sec:corneringcrossing}

In this section we introduce the free cornering of a resource theory, our primary technical device, and show that the free cornering contains special crossing cells with nice formal properties. We begin by recalling the notion of proarrow equipment, specialised to the case of single-object double categories:
\begin{defi}\label{def:equipment}
  Let $\D$ be a single-object double category. $\D$ is called a \emph{proarrow equipment} in case for each $A \in \D_H$ there are distinguished elements $A^\circ$ and $A^\bullet$ of $\D_V$ along with distinguished cells of $\D$:
    \[
      \includegraphics[height=1.7cm,align=c]{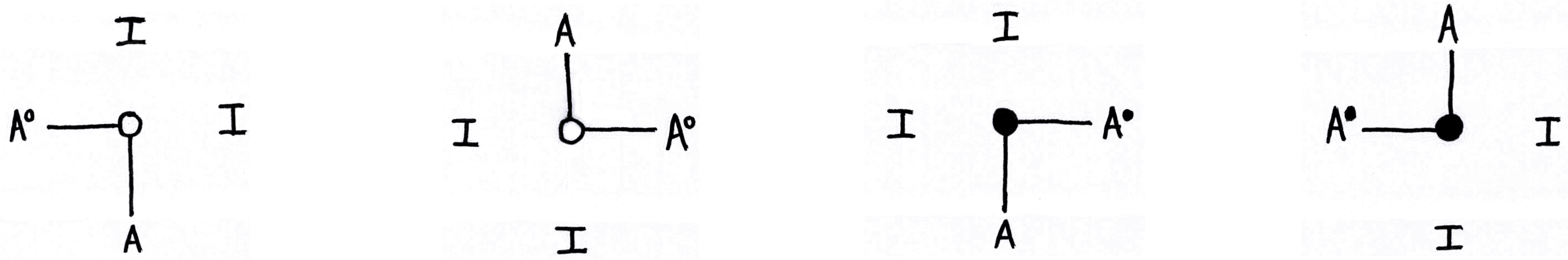}
    \]
    called \emph{$\circ$-corners} and \emph{$\bullet$-corners} respectively, which satisfy the \emph{yanking equations}:
    \[
      \includegraphics[height=1cm,align=c]{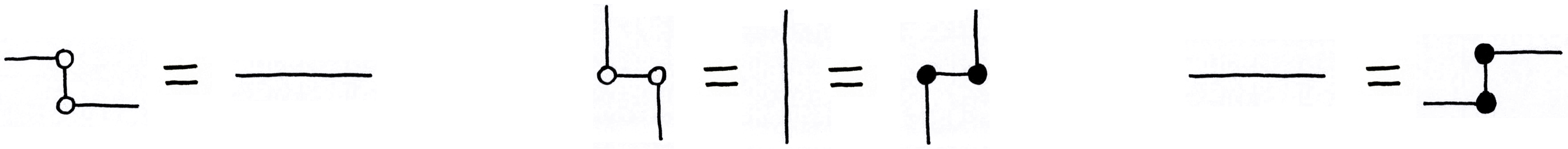}
    \]
\end{defi}

Tersely, the free cornering of a resource theory is the proarrow equipment obtained by freely adding corner cells. Explicitly, we define:
\begin{defi} Let $\A$ be a resource theory. Then the monoid $\ex{\A}$ of \emph{$\A$-valued exchanges} is defined by $\ex{\A} = (\A_0 \times \{\circ,\bullet\})^*$. That is, $\ex{\A}$ is the free monoid on the set $\A_0 \times \{\circ,\bullet\}$ of polarized objects of $\A$, whose elements we write $A^\circ$ and $A^\bullet$. Intuitively, elements of $\ex{\A}$ describe a sequence of resources moving between participants in the exchange, where $A^\circ$ denotes an instance of $A$ moving from left to right, and $A^\bullet$ denotes an instance of $A$ moving from right to left (see Section~\ref{sec:interpretation}).
\end{defi}
\noindent Now the free cornering is given as follows:
\begin{defi} Let $\A$ be a resource theory. Then the \emph{free cornering of $\A$}, written $\corner{\A}$, is the free single-object double category determined by the following data:
\begin{itemize}
  \item The horizontal edge monoid $\corner{\A}_H = (\A_0, \otimes, I)$ is given by the objects of $\A$.
  \item The vertical edge monoid $\corner{\A}_V = \ex{\A}$ is the monoid of $\A$-valued exchanges.
  \item The generating cells consist of corners for each object $A$ of $\A$ as in Definition~\ref{def:equipment}, subject to the yanking equations, along with a vertical cell $\corner{f}$ for each morphism $f : A \to B$ of $\A$ subject to equations as in:
\[
  \includegraphics[height=1.7cm,align=c]{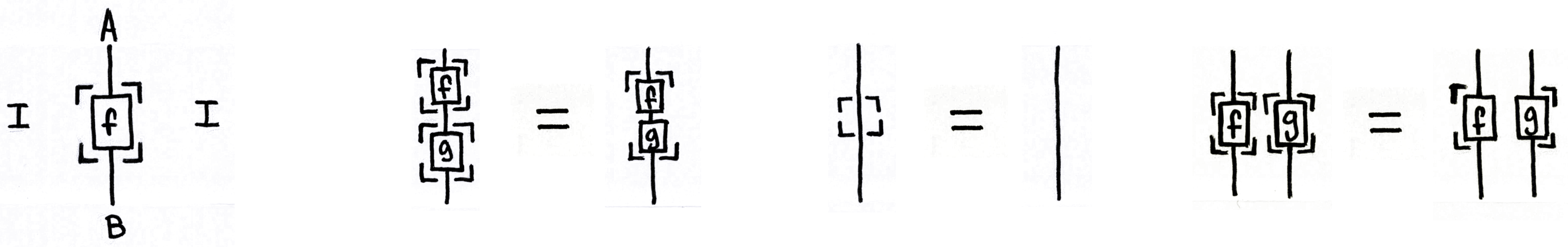}
  \]
\end{itemize}
\end{defi}

\noindent
For a precise development of free double categories see~\cite{Fio08}. In brief: cells are formed from the generating cells by horizontal and vertical composition, subject to the axioms of a double category in addition to any generating equations. We call this the ``free'' cornering both because it is freely generated, and because we imagine there is an adjunction relating proarrow equipments and arbitrary double categories under which $\corner{\A}$ is ``free'' in a more principled sense. We leave the construction of such an adjunction for future work.

An important property of the free cornering is that the vertical cells are the original resource theory:
\begin{prop}\label{prop:verticaloriginal}
There is an isomorphism of categories $\bv\,\corner{\A} \cong \A$.
\end{prop}
\begin{proof}
  Intuitively $\bv\,\corner{\A} \cong \A$ because in a composite vertical cell every wire bent by a corner must eventually be un-bent by the matching corner, which by yanking is the identity. The only other generators are the cells $\corner{f}$, and so any vertical cell in $\corner{A}$ can be written as $\corner{g}$ for some morphism $g$ of $\A$. A more rigorous treatment of corner cells can be found in~\cite{Mye16}, to the same effect.
\end{proof}

Before we properly explain our interest in $\corner{\A}$ we develop a convenient bit of structure: \emph{crossing cells}. For each $B$ of $\corner{\A}_H$ and each $X$ of $\corner{\A}_V$ we define a cell
\[
\includegraphics[height=1.7cm]{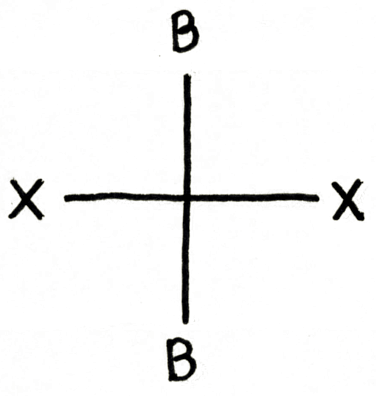}
\]
of $\corner{\A}$ inductively as follows: In the case where $X$ is $A^\circ$ or $A^\bullet$, respectively, define the crossing cell as in the diagrams below on the left and right, respectively:
\[
\includegraphics[height=1.7cm]{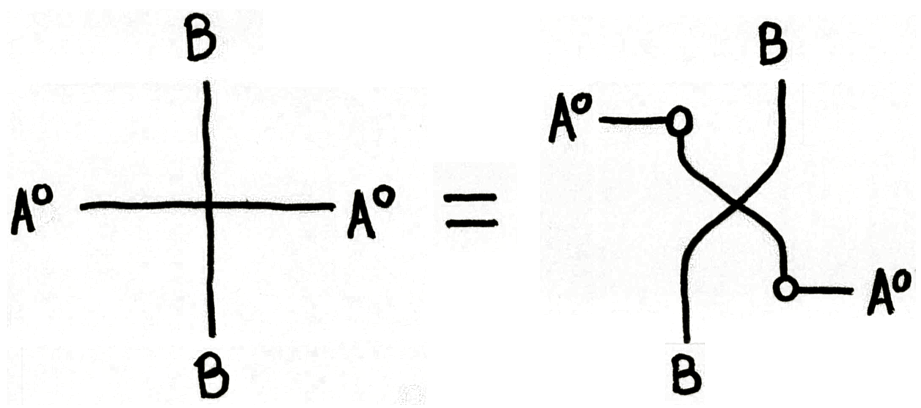}
\hspace{2cm}
\includegraphics[height=1.7cm]{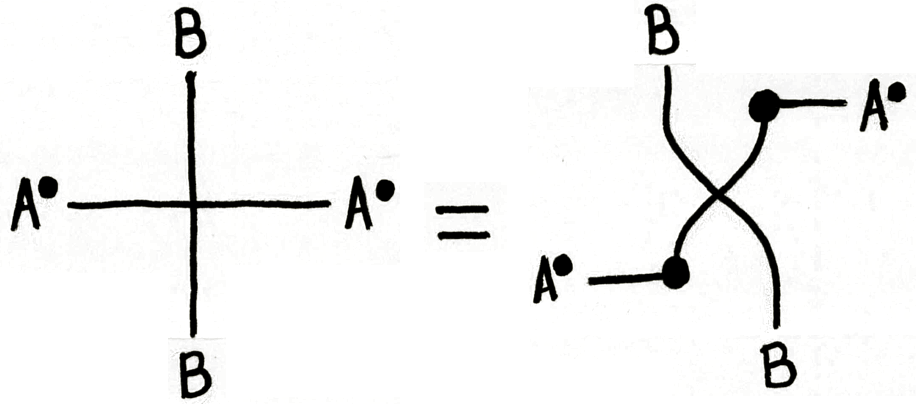}
\]
in the case where $X$ is $I$, define the crossing cell as in the diagram below on the left, and in the composite case define the crossing cell as in the diagram below on the right:
\[
\includegraphics[height=1.7cm]{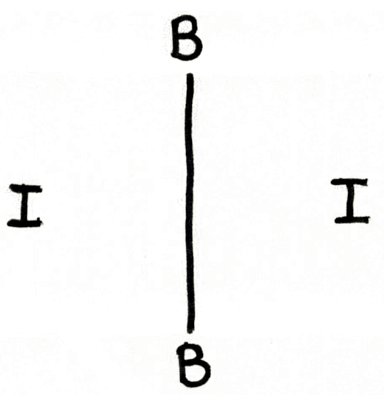}
\hspace{3cm}
\includegraphics[height=1.7cm]{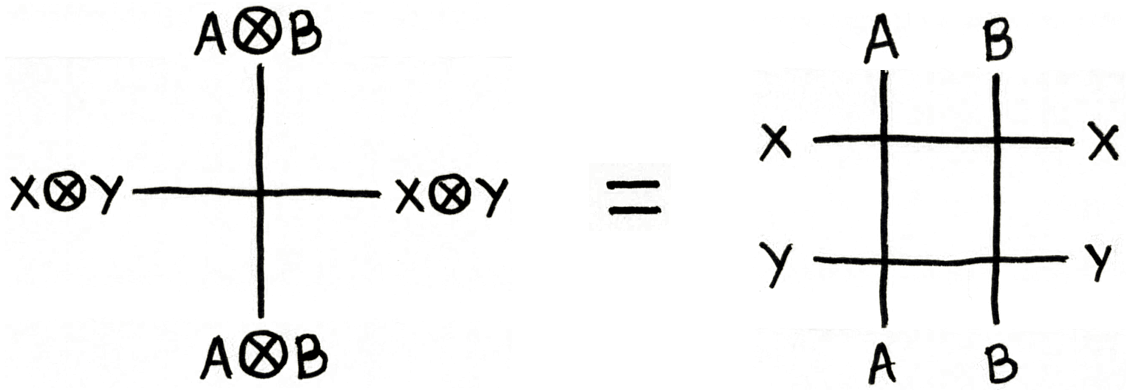}
\]
We prove a technical lemma:
\begin{lem}\label{lem:technicalcrossing}
  For any cell $\alpha$ of $\corner{\A}$ we have
  \[
  \includegraphics[height=1.3cm]{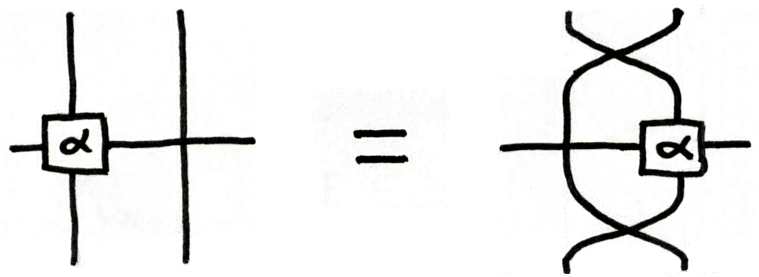}
  \]
\end{lem}
\begin{proof}
  By structural induction on cells of $\corner{\A}$. For the $\circ$-corners we have:
  \[
  \includegraphics[height=1.2cm]{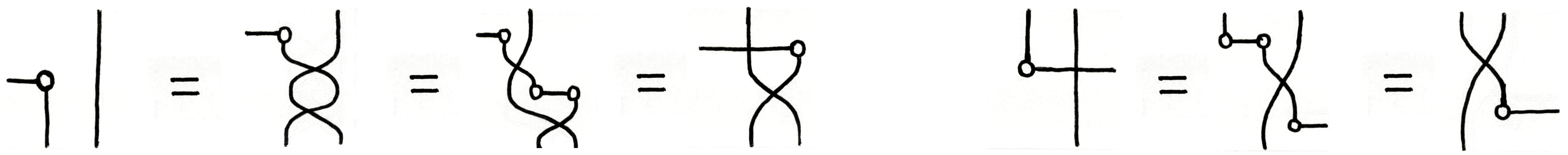}
  \]
  and for the $\bullet$-corners, similarly:
  \[
  \includegraphics[height=1.2cm]{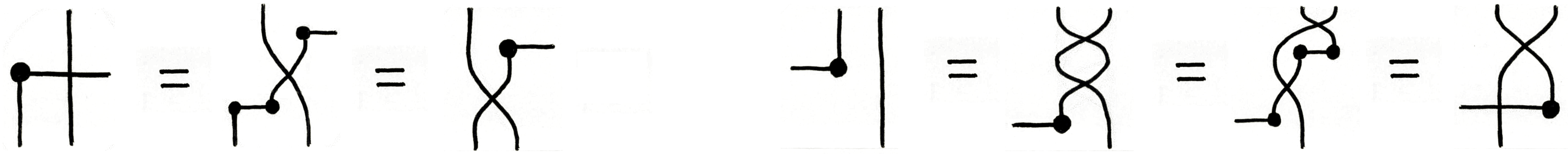}
  \]
  the final base cases are the $\corner{f}$ maps:
  \[
  \includegraphics[height=1.3cm]{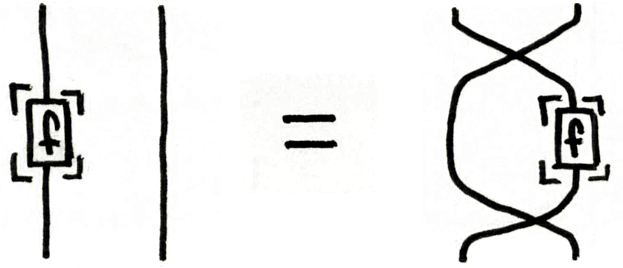}
  \]
  There are two inductive cases. For vertical composition, we have:
  \[
  \includegraphics[height=2.4cm,align=c]{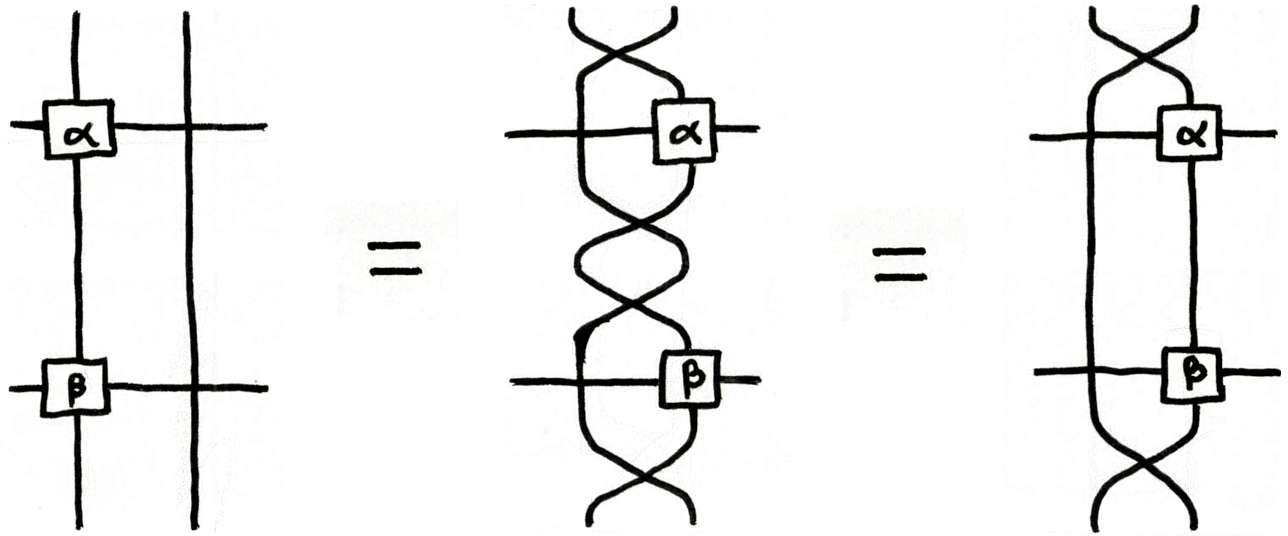}
  \]
  Horizontal composition is similarly straightforward, and the claim follows by induction.
\end{proof}
From this we obtain a ``non-interaction'' property of our crossing cells, similar to the naturality of braiding in symmetric monoidal categories:
\begin{cor} For cells $\alpha$ of $\bv\,\corner{\A}$ and $\beta$ of $\bh\,\corner{\A}$, the following equation holds in $\corner{\A}$:
  \[
  \includegraphics[height=1.5cm]{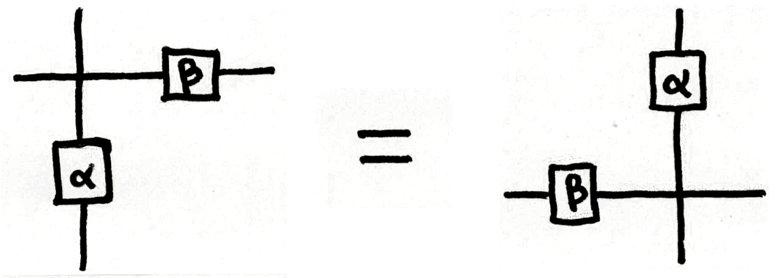}
  \]
\end{cor}

These crossing cells greatly aid in the legibility of diagrams corresponding to cells in $\corner{\A}$, but also tell us something about the categorical structure of $\corner{\A}$, namely that it is a monoidal double category in the sense of~\cite{Shu10}:

\begin{lem}\label{lem:monoidaldouble}
  If $\A$ is a symmetric strict monoidal category then $\corner{\A}$ is a monoidal double category. That is, $\corner{\A}$ is a pseudo-monoid object in the strict 2-category $\bv \mathsf{DblCat}$ of double categories, lax double functors, and vertical transformations.
\end{lem}
\begin{proof}
  We give the action of the tensor product on cells:
  \[
    \includegraphics[height=1.7cm]{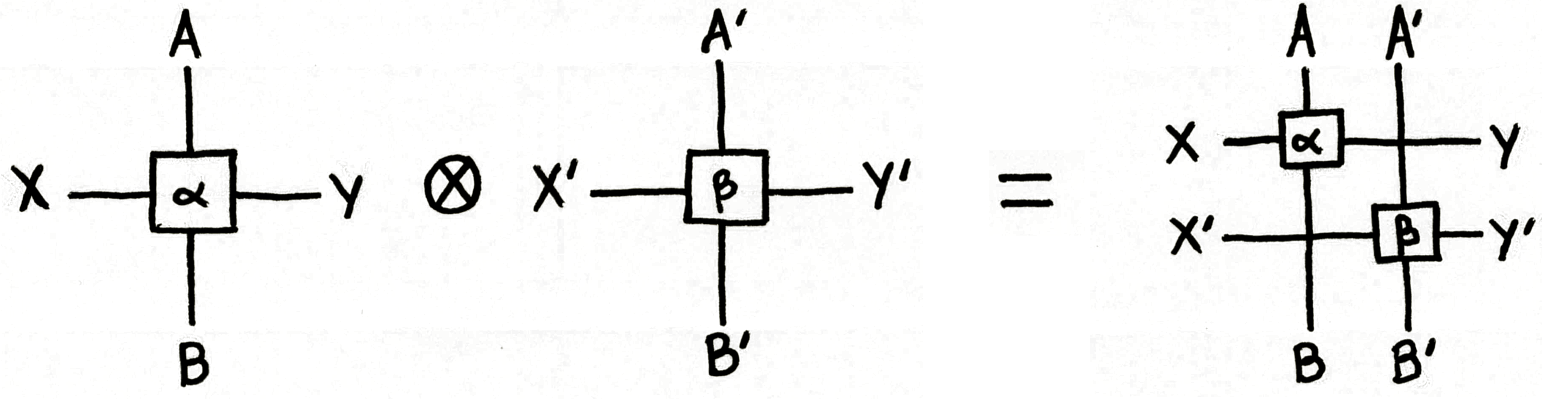}
  \]
  This defines a pseudofunctor, with the component of the required vertical transformation given by exchanging the two middle wires as in:
  \[
    \includegraphics[height=2.5cm]{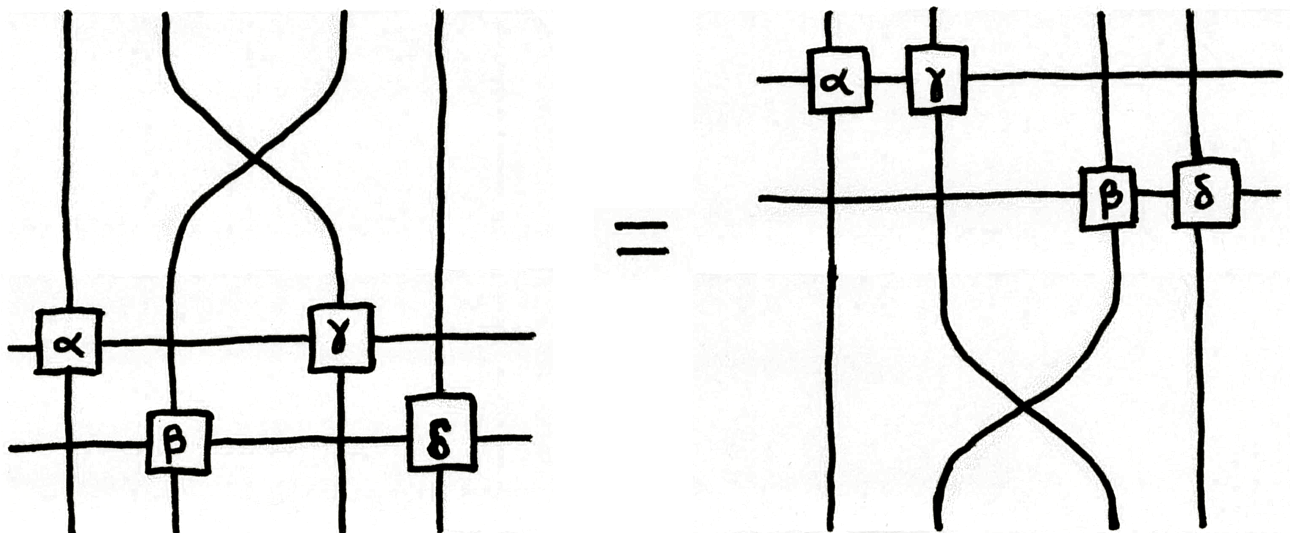}
  \]
  Notice that $\otimes$ is strictly associative and unital, in spite of being only pseudo-functorial.
\end{proof}

\section{Concurrency Through Cornering}\label{sec:interpretation}
We proceed to extend the resource-theoretic interpretation of some symmetric strict monoidal category $\A$ to its free cornering $\corner{\A}$. We interpret elements of $\corner{\A}_V = \ex{\A}$ as \emph{$\A$-valued exchanges}. Each exchange $X_1 \otimes \cdots \otimes X_n$ involves a left participant and a right participant giving each other resources in sequence, with $A^\circ$ indicating that the left participant should give the right participant an instance of $A$, and $A^\bullet$ indicating the opposite. For example say the left participant is $\alice$ and the right participant is $\bob$. Then we can picture the exchange $A^\circ \otimes B^\bullet \otimes C^\bullet$ as:
\[
\alice \rightsquigarrow
\includegraphics[height=1.3cm,align=c]{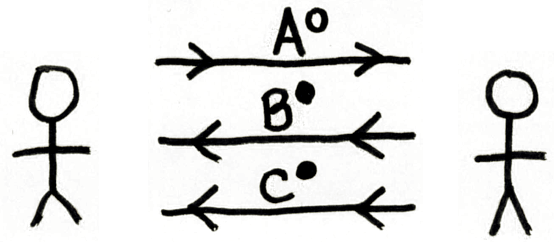}
\leftsquigarrow \bob
\]
Think of these exchanges as happening \emph{in order}. For example the exchange pictured above demands that first $\alice$ gives $\bob$ an instance of $A$, then $\bob$ gives $\alice$ an instance of $B$, and then finally $\bob$ gives $\alice$ an instance of $C$.

We interpret cells of $\corner{\A}$ as \emph{concurrent transformations}. Each cell describes a way to transform the collection of resources given by the top boundary into that given by the bottom boundary, via participating in $\A$-valued exchanges along the left and right boundaries. For example, consider the following cells of $\corner{\B}$:
\[
\includegraphics[height=2.7cm]{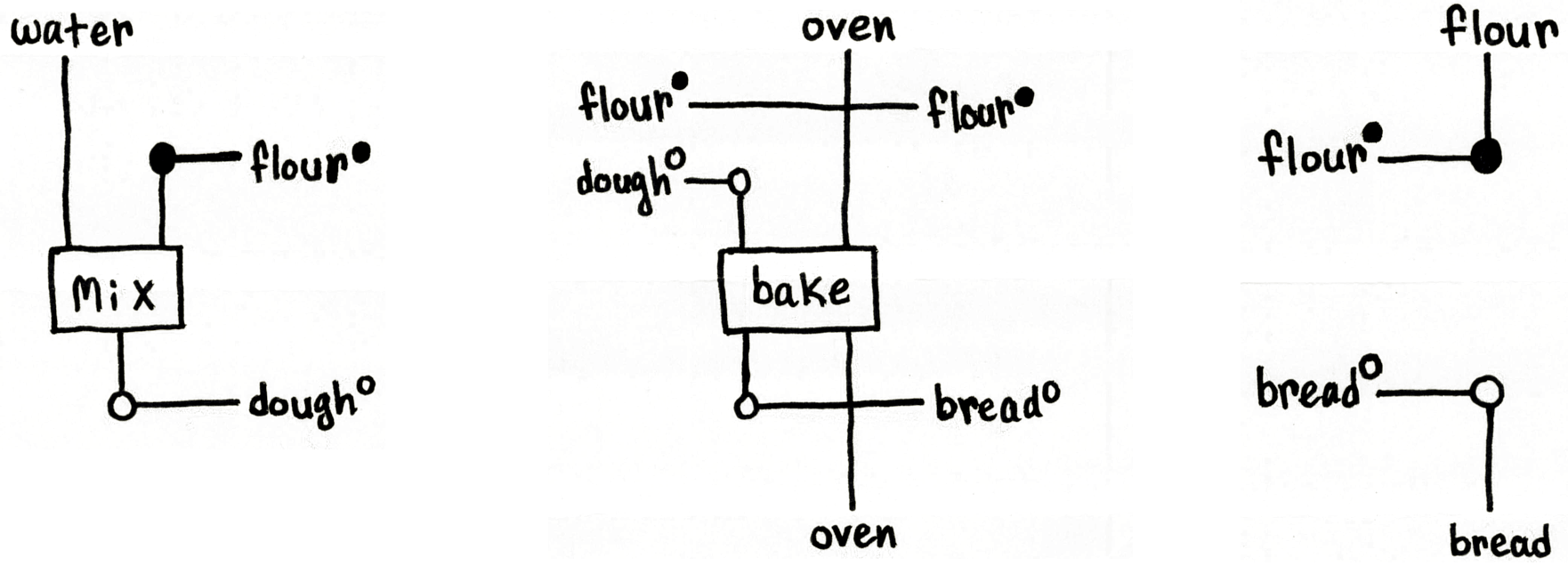}
\]
From left to right, these describe: A procedure for transforming \texttt{water} into nothing by \texttt{mix}ing it with \texttt{flour} obtained by exchange along the right boundary, then sending the resulting \texttt{dough} away along the right boundary; A procedure for transforming an \texttt{oven} into an \texttt{oven}, receiving \texttt{flour} along the right boundary and sending it out the left boundary, then receiving \texttt{dough} along the left boundary, which is \texttt{bake}d in the \texttt{oven}, with the resulting \texttt{bread} sent out along the right boundary; Finally, a procedure for turning \texttt{flour} into \texttt{bread} by giving it away and then receiving \texttt{bread} along the left boundary.
When we compose these concurrent transformations horizontally in the evident way, they give a transformation of resources in the usual sense, i.e.,  a morphism of $\A \cong \bv\,\corner{\A}$:
\[
\includegraphics[height=2.7cm]{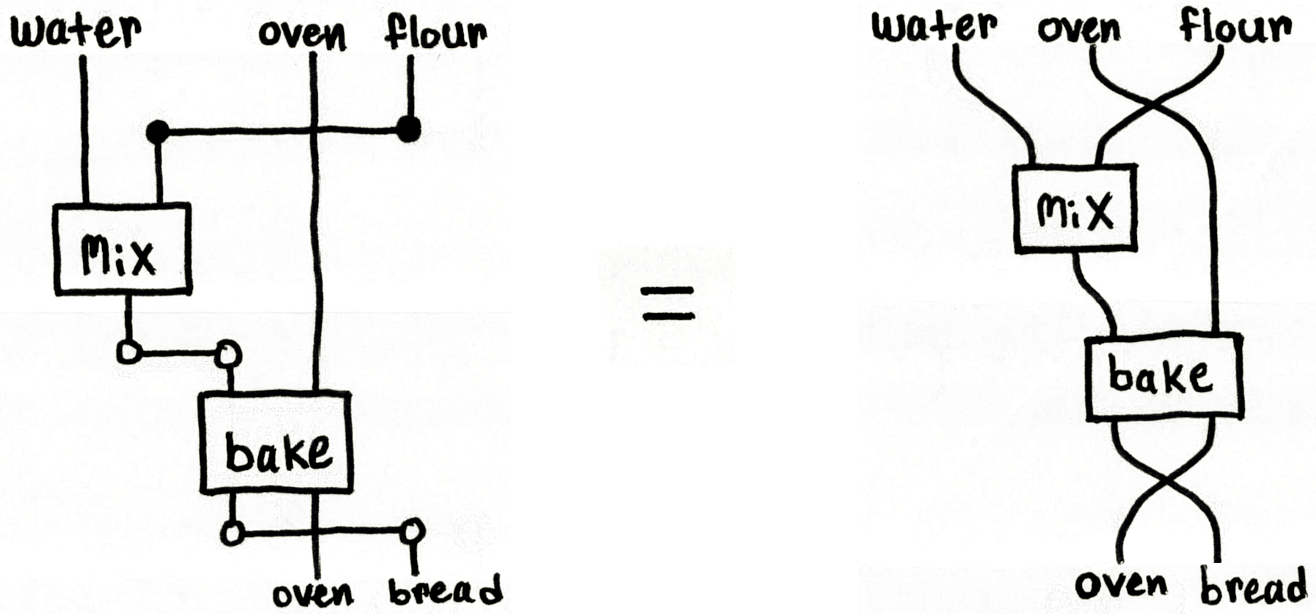}
\]

We understand equality of cells in $\corner{\A}$ much as we understand equality of morphisms in a resource theory: two cells should be equal in case the transformations they describe would have the same effect on the resources involved. In this way, cells of $\corner{\A}$ allow us to break a transformation into many concurrent parts. Note that with the crossing cells, it is possible for cells that are not immediately adjacent to exchange resource across the cells in between them. In the above example, $\texttt{flour}$ is sent from the rightmost cell to the leftmost cell across the middle cell. This makes the double-categorical structure less constraining that it may seem at first. For example we might rearrange our previous example into the following horizontally composable cells of $\corner{\mathfrak{B}}$:
\begin{mathpar}
  \includegraphics[height=1.7cm,align=c]{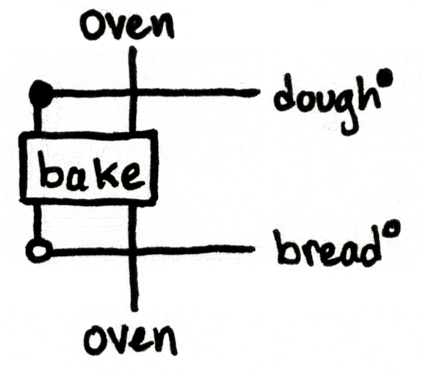}

  \includegraphics[height=1.7cm,align=c]{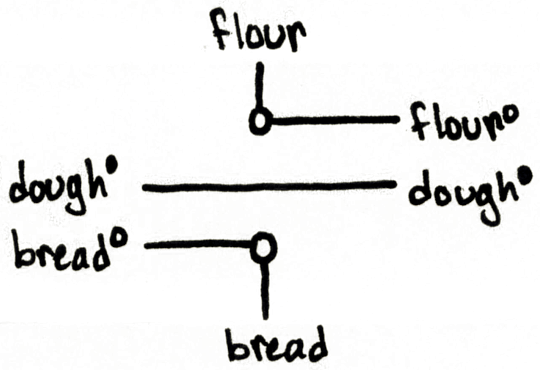}

  \includegraphics[height=1.7cm,align=c]{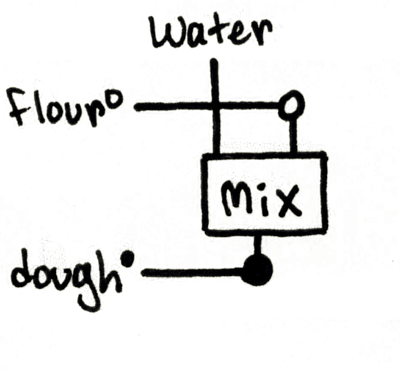}
\end{mathpar}
When composed, we obtain a similar morphism of $\A$:
\[
\includegraphics[height=2.8cm,align=c]{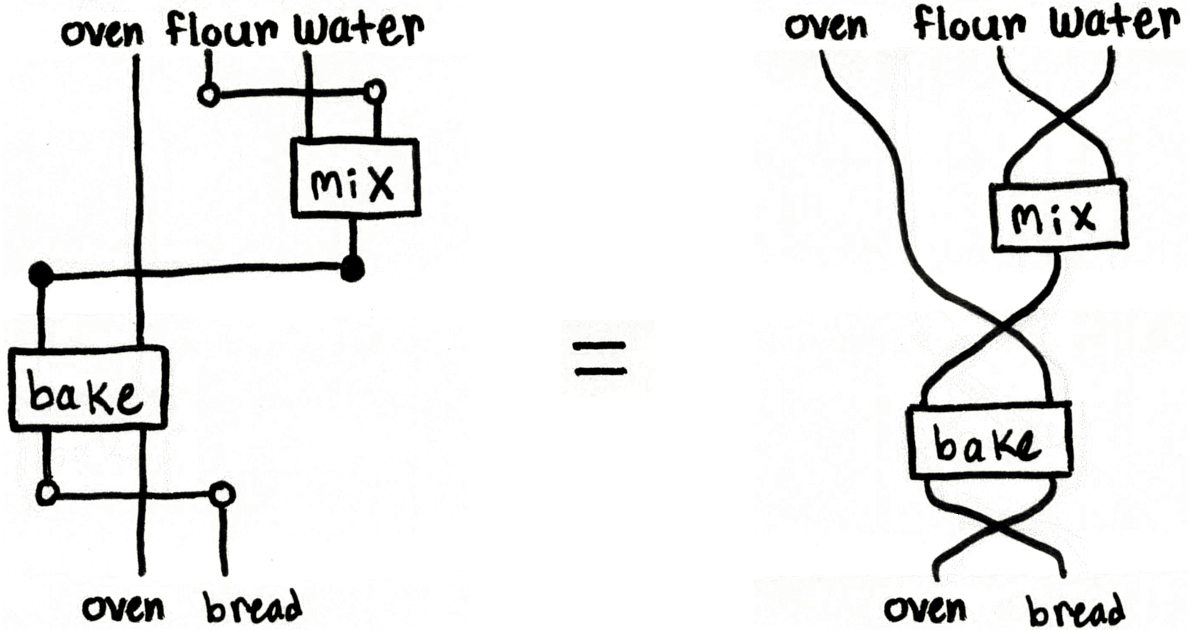}
\]
It is worth mentioning that the difference between $\texttt{oven} \otimes \texttt{flour} \otimes \texttt{water}$ and $\texttt{water} \otimes \texttt{oven} \otimes \texttt{flour}$ is negligible since any permutation of a collection of resources is naturally isomorphic to the original collection as an object of $\A$.

\section{Horizontal Cells as Resource Transducers}\label{sec:transducers}

If $\A$ is a resource theory, then the category $\bh\,\corner{\A}$ of horizontal cells of the free cornering can be understood as a category of ($\A$-valued) \emph{resource transducers}.\footnote{The word ``transducer'' is derived from the latin words \emph{trans} --- meaning ``across'' and \emph{ducere} --- meaning ``lead''. We feel this is a good fit for the horizontal cells of the free cornering, which can be understood as a method of leading resources across the cell in question.} Specifically, recall our interpretation of $\ex{\A} = (\bh\,\corner{\A})_0$ as $\A$-valued exchanges, in which two parties $\texttt{Alice}$ and $\texttt{Bob}$ must supply or retreive the resources involved in the exchange in the order specified, with who gives whom what determined by the polarity of the resources (see Section~\ref{sec:interpretation}). Let $h : X \to Y$ be an arrow of $\bh\,\corner{\A}$. We can understand $h$ as a machine operated by a left and right participant, again called \texttt{Alice} and \texttt{Bob} respectively. To operate the machine, \texttt{Alice} must play the left hand role of the domain exchange $X$ and \texttt{Bob} must play the right hand role of the codomain exchange $Y$. The morphism $h$ describes the internals of the machine. For example, consider the following morphism of $\bh\,\corner{\A}$:
\[
\alice \rightsquigarrow
\includegraphics[height=1.5cm,align=c]{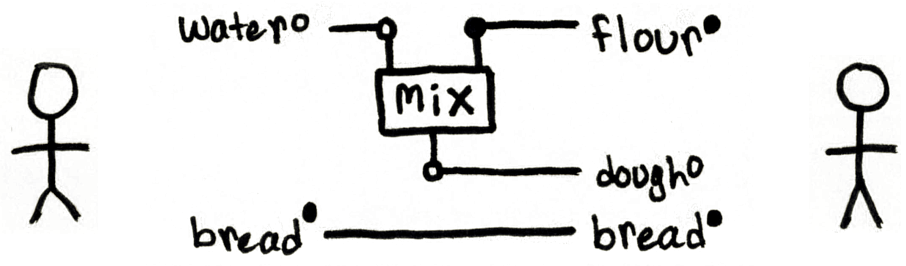}
\leftsquigarrow \bob
\]
To operate the transducer, \texttt{Alice} must supply water and then receive bread, while \texttt{Bob} must supply flour, receive dough, and then supply bread. The effect of the machine is to mix the flour and water initially supplied into the dough \texttt{Bob} receives, and then to send the bread \texttt{Bob} supplies to \texttt{Alice}.

The transducer interpretation (along with our previous interpretation of the whole of $\corner{\A}$) makes $\bh\,\corner{\A}$ into a category of independent interest, and in this section we will study it. Compounding our interest is the fact that $\bh\,\corner{\A}$ is rather unusual. It is of course a monoidal category (see Section~\ref{sec:singledouble}) but fails to have any of the properties common to monoidal categories. Selinger's survey paper~\cite{Sel10} lists many such properties, for example:
\begin{defiC}[\cite{Sel10}]
  A monoidal category is \emph{spatial} in case for all objects $X$ and arrows $h : I \to I$ we have:
  \[
  \includegraphics[height=1.2cm]{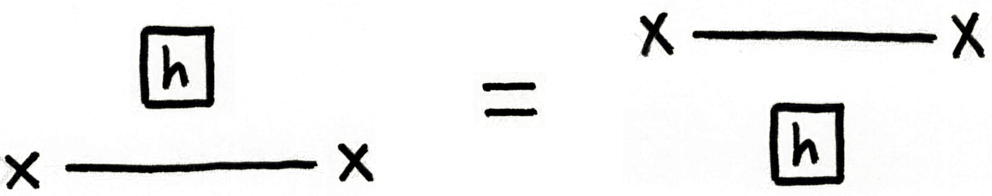}
  \]
\end{defiC}
\noindent It is easy to see that $\bh\,\corner{\A}$ has the property of being spatial:
\begin{lem}\label{lem:spatial}
  $\bh\,\corner{\A}$ is spatial.
\end{lem}
\begin{proof}
  We use the fact that every symmetric monoidal category is spatial. The proof is by induction on the type $X$ of the wire. If $X$ is $A^\circ$ we have:
  \[
  \includegraphics[height=1.2cm]{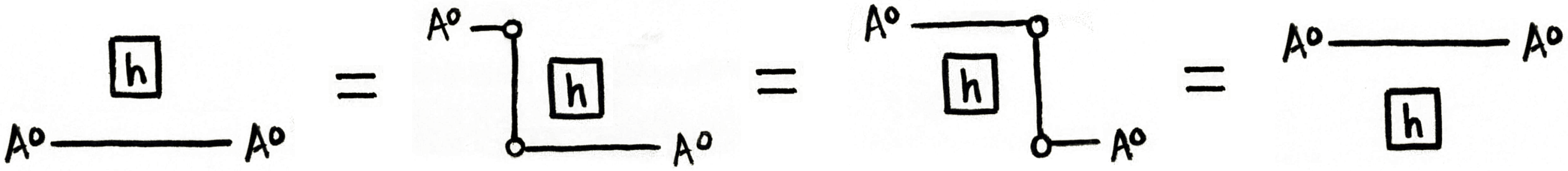}
  \]
  and so the spatial axiom holds. Similarly the spatial axiom holds if $X$ is $A^\bullet$. If $X$ is $I$ the spatial axiom holds trivially, and the inductive case is immediate.
\end{proof}
\noindent We note that $\bh\,\corner{\A}$ has \emph{no other property} found in the aforementioned survey paper.

Much of the structure that $\bh\,\corner{\A}$ does have consists of isomorphisms formed of corner cells. While isomorphic objects in $\bv\,\corner{\A} \cong \A$ can be thought of as equivalent collections of resources --- being freely transformable into each other --- we understand isomorphic objects in $\bh\,\corner{\A}$ as \emph{equivalent exchanges}. For example, there are many ways for $\alice$ to give $\bob$ an $A$ and a $B$: Simultaneously, as $A \otimes B$; one after the other, as $A$ and then $B$; or in the other order, as $B$ and then $A$. While these are different sequences of events, they achieve the same thing, and are thus equivalent. Similarly, for $\alice$ to give $\bob$ an instance of $I$ is equivalent to nobody doing anything. Formally, we have:
\begin{lem}\label{lem:moraliso}
  In $\bh\,\corner{\A}$ we have for any $A,B$ of $\A$:
  \begin{enumerate}
  \item $I^\circ \cong I \cong I^\bullet$.
  \item $A^\circ \otimes B^\circ \cong B^\circ \otimes A^\circ$ and $A^\bullet \otimes B^\bullet \cong B^\bullet \otimes A^\bullet$.
  \item $(A \otimes B)^\circ \cong A^\circ \otimes B^\circ$ and $(A \otimes B)^\bullet \cong A^\bullet \otimes B^\bullet$
  \end{enumerate}
\end{lem}
\begin{proof}
  \begin{enumerate}
  \item For $I \cong I^\circ$, consider the $\circ$-corners corresponding to $I$:
    \[
    \includegraphics[height=1.7cm]{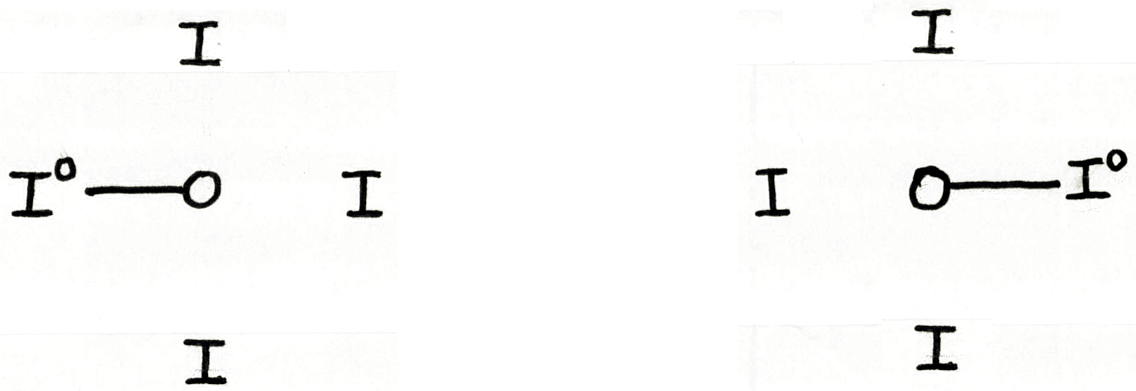}
    \]
    we know that these satisfy the yanking equations:
    \[
    \includegraphics[height=1cm]{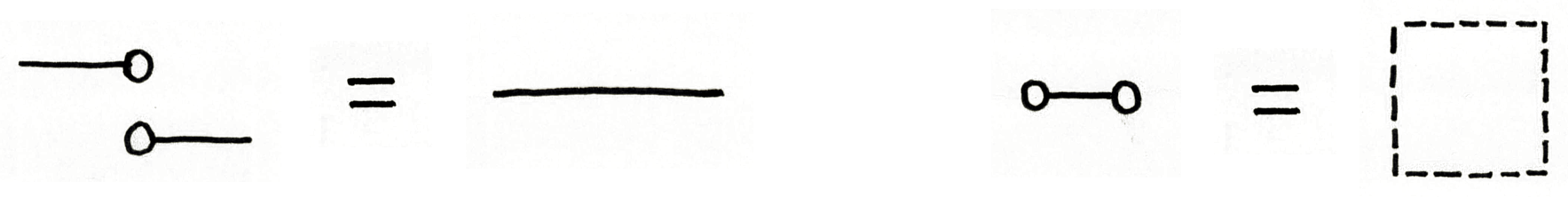}
    \]
    which exhibits an isomorphism $I \cong I^\circ$. Similarly, $I \cong I^\bullet$. Thus, we see formally that exchanging nothing is the same as doing nothing.
  \item The $\circ$-corner case is the interesting one: Define the components of our isomorphism to be:
    \[
    \includegraphics[height=0.8cm,align=c]{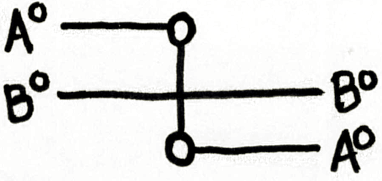}
    \hspace{0.5cm}
    \text{and}
    \hspace{0.5cm}
    \includegraphics[height=0.8cm,align=c]{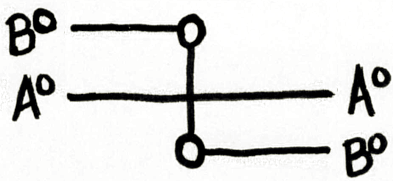}
    \]
    then for both of the required composites we have:
    \[
    \includegraphics[height=2cm]{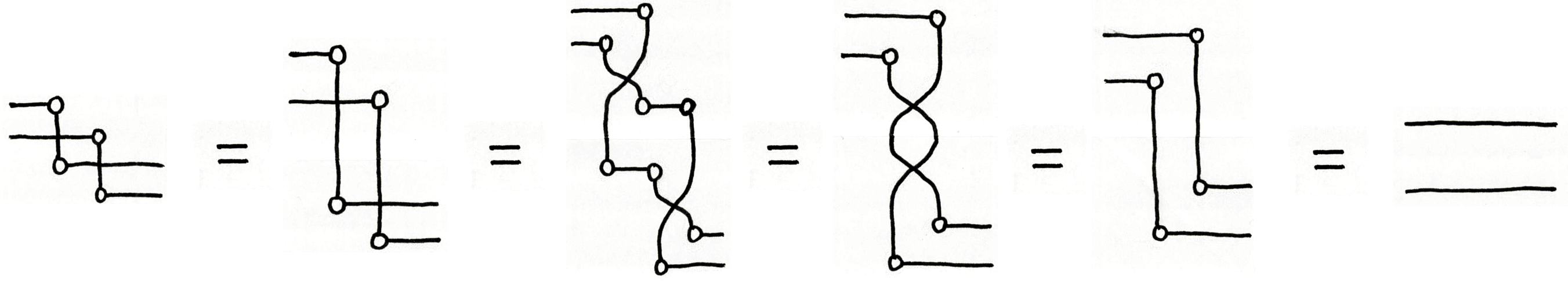}
    \]
    and so $A^\circ \otimes B^\circ \cong B^\circ \otimes A^\circ$. Similarly $A^\bullet \otimes B^\bullet \cong B^\bullet \otimes A^\bullet$. This captures formally the fact that if $\alice$ is going to give $\bob$ an $A$ and a $B$, it doesn't really matter which order she does it in.
  \item Here it is convenient to switch between depicting a single wire of sort $A \otimes B$ and two wires of sort $A$ and $B$ respectively in our string diagrams. To this end, we allow ourselves to depict the identity on $A \otimes B$ in multiple ways, using the notation of~\cite{Coc17}:
    \[
    \includegraphics[height=1.7cm]{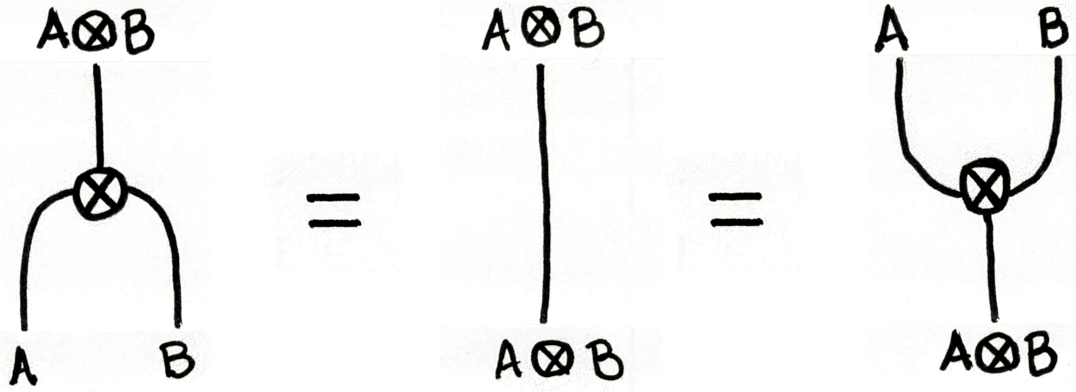}
    \]
    Then the components of our isomorphism $(A \otimes B)^\circ \cong A^\circ \otimes B^\circ$ are:
    \[
    \includegraphics[height=1.2cm,align=c]{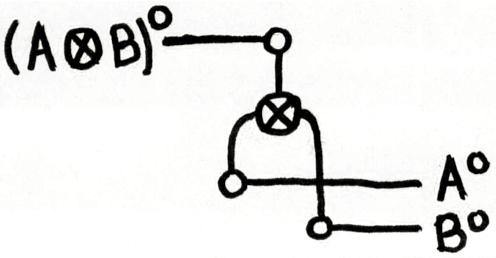}
    \hspace{0.5cm}
    \text{and}
    \hspace{0.5cm}
    \includegraphics[height=1.2cm,align=c]{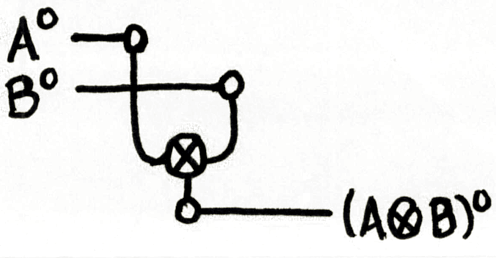}
    \]
    and, much as in (ii), it is easy to see that the two possible composites are both identity maps. Similarly, $(A \otimes B)^\bullet \cong (A^\bullet \otimes B^\bullet)$. This captures formally the fact that giving away a collection is the same thing as giving away its components.
    \qedhere
  \end{enumerate}
\end{proof}

\noindent
For example, we should be able to compose the cells on the left and right below horizontally, since their right and left boundaries, respectively, indicate equivalent exchanges:
\[
\includegraphics[height=1.5cm,align=c]{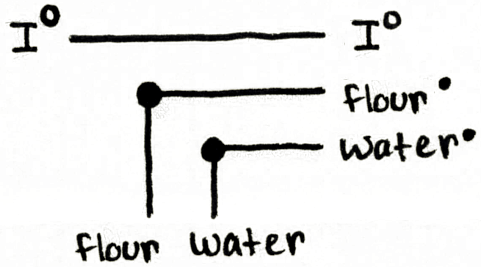}
\hspace{1.5cm}
\includegraphics[height=1cm,align=c]{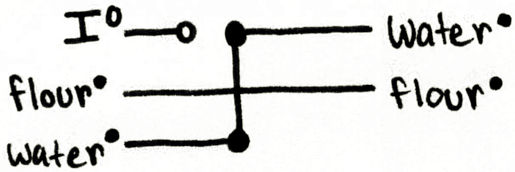}
\hspace{1.5cm}
\includegraphics[height=1.3cm,align=c]{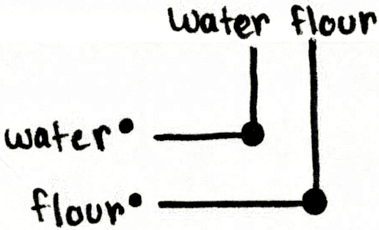}
\]
Our lemma tells us that in cases like this there will be a mediating isomorphism, as above in the middle, making composition possible.

It is worth noting that we \emph{do not} have $A^\circ \otimes B^\bullet \cong B^\bullet \otimes A^\circ$:
\begin{observation}\label{obs:halfsym}
  There is a morphism $d^\circ_\bullet : A^\circ \otimes B^\bullet \to B^\bullet \otimes A^\circ$ in one direction, defined by
\[
\includegraphics[height=1cm]{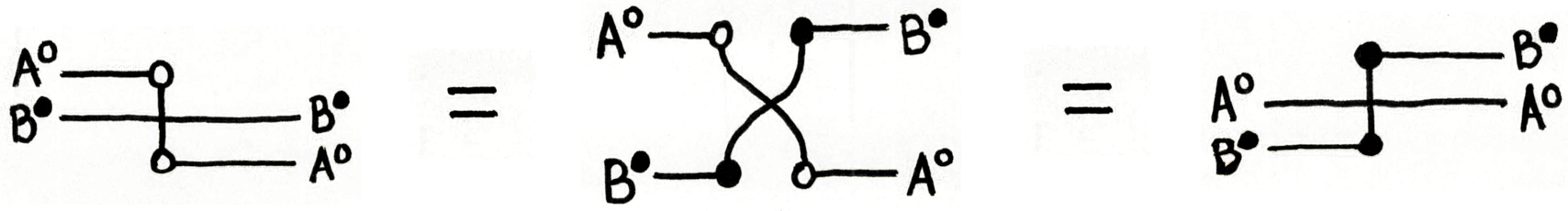}
\]
but there need not be a morphism in the other direction, and this is not in general invertible. In particular, $\bh\,\corner{\A}$ is monoidal, but need not be symmetric.
\end{observation}

This observation reflects formally the intuition that if I receive some resources before I am required to send any, then I can send some of the resources that I receive. However, if I must send the resources first, this is not the case. In this way, $\bh\,\corner{\A}$ contains a sort of causal structure.

Next, we find that $\bh\,\corner{\A}$  contains the original resource theory $\A$ as a subcategory in two different ways, one for each polarity:
\begin{lem}\label{lem:strongmonoidal}
  There are strong monoidal functors $(-)^\circ : \A \to \bh\,\corner{\A}$ and $(-)^\bullet : \A^\op \to \bh\,\corner{\A}$ defined respectively on $f : A \to B$ of $\A$ by:
\begin{mathpar}
    \includegraphics[height=1.2cm,align=c]{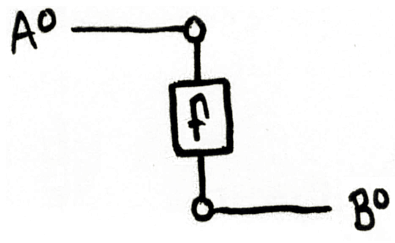}

    \text{and}

    \includegraphics[height=1.2cm,align=c]{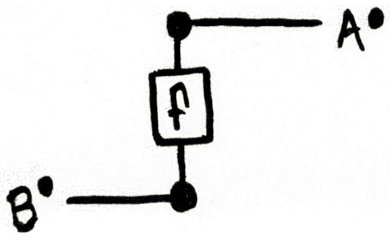}
\end{mathpar}
  Further, each of these functors is full and faithful.
\end{lem}
\begin{proof}
  $(-)^\circ$ is functorial as in:
  \[
  \includegraphics[height=1.5cm,align=c]{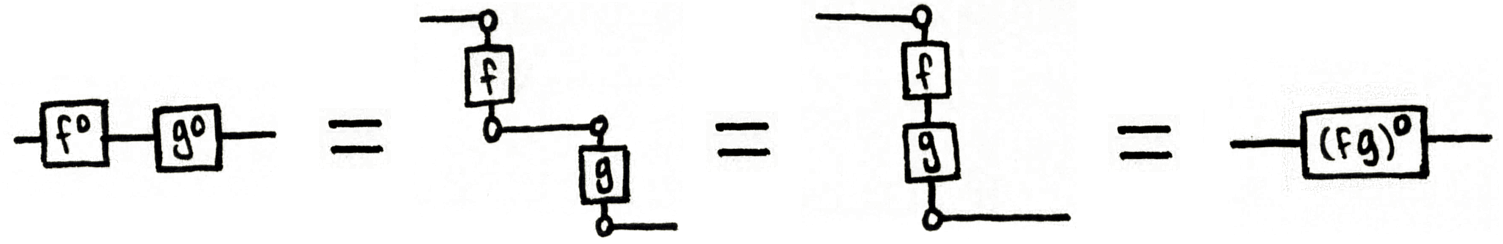}
  \]
  It interacts with the tensor product in $\A$ as in:
  \[
  \includegraphics[height=1.5cm,align=c]{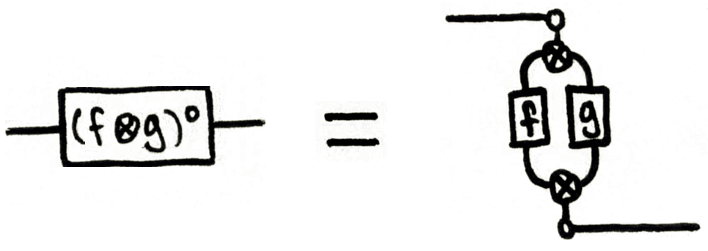}
  \]
  and is therefore strong monoidal as a consequence of Lemma~\ref{lem:moraliso}. Further $(-)^\circ$ is faithful because $\corner{\A}$ is freely generated. It is full because of the coherence theorem of~\cite{Mye16}, which implies that for any horizontal cell (morphism of $\bh\,\corner{\A}$) $h : A^\circ \to B^\circ$ we may yank all of the wires straight to obtain an equal morphism $f^\circ = h$ for some $f : A \to B$ of $\A$. Similarly, $(-)^\bullet$ is functorial, strong monoidal, full, and faithful.
\end{proof}

There is also a contravariant involution $(-)^* : \bh\,\corner{\A}^\op \to \bh\,\corner{\A}$. As an intermediate step we define an operation on the cells of $\corner{\A}$ as follows: For $A \in \A_0 = \bv\,\corner{\A}$ let $A^* = A$. For $X \in \ex{\A} = \bh\,\corner{\A}$ define $X^*$ inductively: $I^* = I$, $(A^\circ)^* = A^\bullet$, $(A^\bullet)^* = A^\circ$, and $(X \otimes Y)^* = X^* \otimes Y^*$. On cells of $\corner{\A}$ we also define $(-)^*$ inductively: The base cases are $\corner{f}^* = \corner{f}$ along with:
\begin{mathpar}
  {
    \includegraphics[height=1.7cm,align=c]{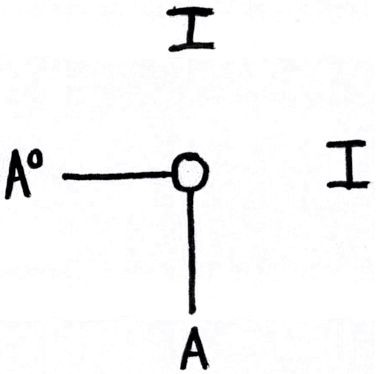}
    \hspace{0.5cm}\stackrel{*}{\mapsto}\hspace{0.5cm}
    \includegraphics[height=1.7cm,align=c]{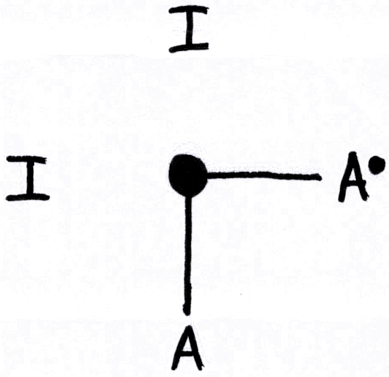}
  }

  {
    \includegraphics[height=1.7cm,align=c]{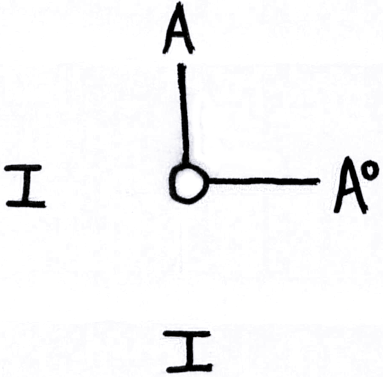}
    \hspace{0.5cm}\stackrel{*}{\mapsto}\hspace{0.5cm}
    \includegraphics[height=1.7cm,align=c]{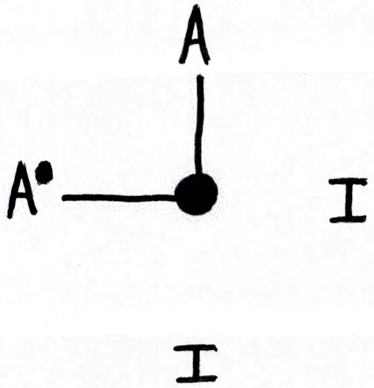}
  }

  {
    \includegraphics[height=1.7cm,align=c]{figs/involution-rhs-one.png}
    \hspace{0.5cm}\stackrel{*}{\mapsto}\hspace{0.5cm}
    \includegraphics[height=1.7cm,align=c]{figs/involution-lhs-one.png}
  }

  {
    \includegraphics[height=1.7cm,align=c]{figs/involution-lhs-two.png}
    \hspace{0.5cm}\stackrel{*}{\mapsto}\hspace{0.5cm}
    \includegraphics[height=1.7cm,align=c]{figs/involution-rhs-two.png}
  }
\end{mathpar}
and the inductive cases are:
\begin{mathpar}
  {
    \includegraphics[height=1.7cm,align=c]{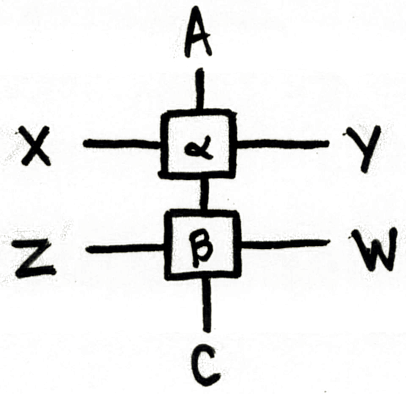}
    \hspace{0.5cm}\stackrel{*}{\mapsto}\hspace{0.5cm}
    \includegraphics[height=1.7cm,align=c]{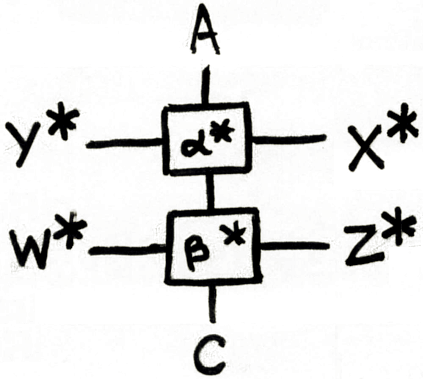}
  }

  {
    \includegraphics[height=1.7cm,align=c]{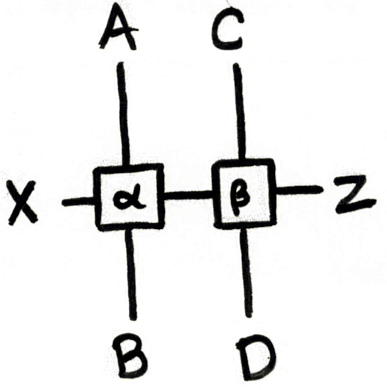}
    \hspace{0.5cm}\stackrel{*}{\mapsto}\hspace{0.5cm}
    \includegraphics[height=1.7cm,align=c]{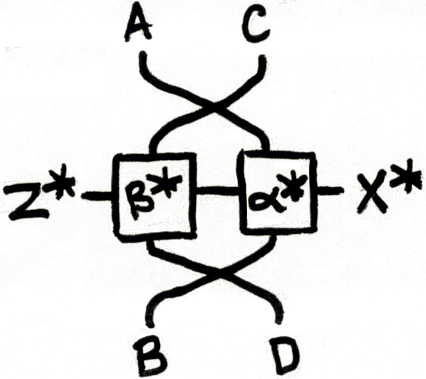}
  }
\end{mathpar}
Informally, $\alpha^*$ is the mirror image of $\alpha$. It is easy to see that we have $\alpha^{**} = \alpha$ for any cell $\alpha$ of $\corner{\A}$. Thus, restricting $(-)^*$ to $\bh\,\corner{\A}$ gives:
\begin{lem}\label{lem:involution}
  There is a contravariant involution $(-)^* : \bh\,\corner{\A}^\op \to \bh\,\corner{\A}$ with the property that $(f \otimes g)^* = f^* \otimes g^*$.\footnote{It is tempting to call this a \emph{contravariant monoidal involution}, but in the covariant case a \emph{monoidal involution} $(-)^\iota$ has the property that $(f \otimes g)^\iota = g^\iota \otimes f^\iota$, twisting the tensor product~\cite{Egg11}. We refrain from coining any new technical terms lest a ``contravariant monoidal involution'' turn out to be better suited to describing contravariant involutions that twist the tensor product instead of those that do not.}
\end{lem}

We discuss one final bit of structure in $\bh\,\corner{\A}$, concerning the following arrows:
\begin{mathpar}
\includegraphics[height=1cm,align=c]{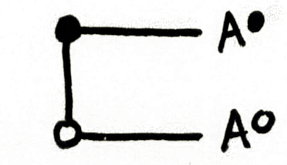}

\includegraphics[height=1cm,align=c]{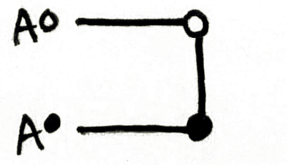}
\end{mathpar}
\noindent These are reminiscent of the string diagrams for rigid monoidal categories, these arrows make $A^\circ$ into the left dual of $A^\bullet$ (and so make $A^\bullet$ into the right dual of $A^\circ$). However, $\bh\,\corner{\A}$ is neither left nor right rigid: for example $A^\circ \otimes B^\bullet$ has neither a left nor right dual.
It is natural to ask whether the arrows introduced above carry significant categorical structure. We give one answer, and in doing so connect the present work to Cockett and Pastro's logic of message passing~\cite{Coc09}. In particular, the categorical semantics of this logic of message passing is given by \emph{linear actegories}. If $\A$ is a symmetric monoidal category, a linear $\A$-actegory is given by a linearly distributive category $\X$ (see e.g.,~\cite{Coc17}) together with two functors:
\begin{mathpar}
\circ : \A \times \X \to \X

\bullet : \A^\op \times \X \to \X
\end{mathpar}
such that $\circ$ is the paramaterised left adjoint of $\bullet$ --- that is, for all $A \in \A_0$ we have $A \circ - \dashv A \bullet -$ --- along with nine natural families of arrows subject to a large number of coherence conditions.

The category $\bh\,\corner{\A}$ exhibits similar, if much simpler, structure. In particular the strong monoidal functors $(-)^\circ$ and $(-)^\bullet$ of Lemma~\ref{lem:strongmonoidal} allow us to define $\circ : \A \times \bh\,\corner{\A} \to \bh\,\corner{\A}$ and $\bullet : \A^\op \times \bh\,\corner{\A} \to \bh\,\corner{\A}$ by $f \circ h = f^\circ \otimes h$ and $f \bullet h = f^\bullet \otimes h$. Echoing the definition of a linear actegory, we have:
\begin{lem}
  $\circ$ is the paramaterised left adjoint of $\bullet$. That is, for all $A \in \A$ the functors $A \circ - : \bh\,\corner{\A} \to \bh\,\corner{\A}$ and $A \bullet - : \bh\,\corner{\A} \to \bh\,\corner{\A}$ defined on $h : X \to Y$ by, respectively:
\begin{mathpar}
  \includegraphics[height=1.5cm,align=c]{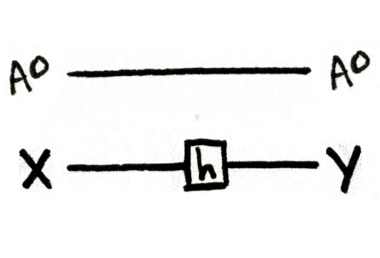}

  \text{and}

  \includegraphics[height=1.5cm,align=c]{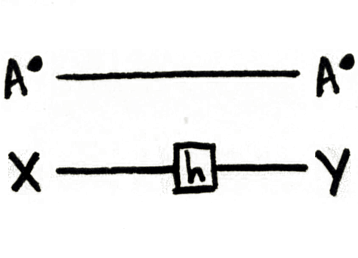}
\end{mathpar}
are such that $A \circ - \dashv A \bullet -$.
\end{lem}
\begin{proof}
  Fix an object $A \in \A$. We require natural families of morphisms $\eta_{A,X} : X \to A \bullet (A \circ X)$ and $\varepsilon_{A,X} : A \circ (A \bullet X) \to X$ in $\bh\,\corner{\A}$ that satisfy the triangle identities. Define $\eta_{A,X}$ and $\varepsilon_{A,X}$, respectively, by
\begin{mathpar}
  \includegraphics[height=1.5cm,align=c]{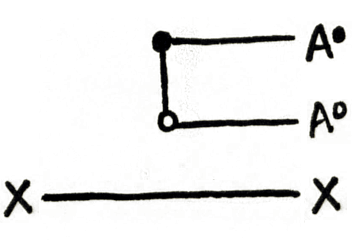}

  \text{and}

  \includegraphics[height=1.5cm,align=c]{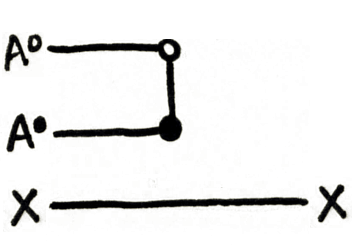}
\end{mathpar}
Now the triangle identities hold by repeated yanking, as in:
\begin{mathpar}
  \includegraphics[height=1.5cm,align=c]{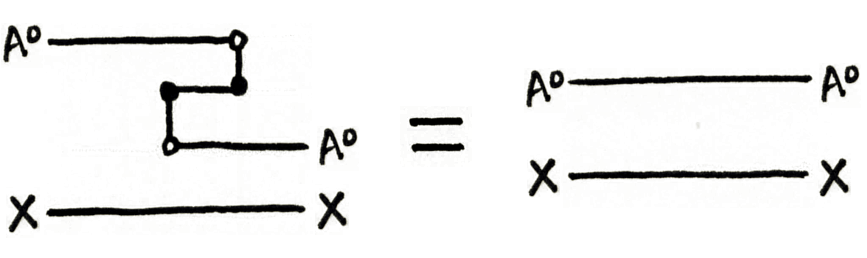}

  \text{and}

  \includegraphics[height=1.5cm,align=c]{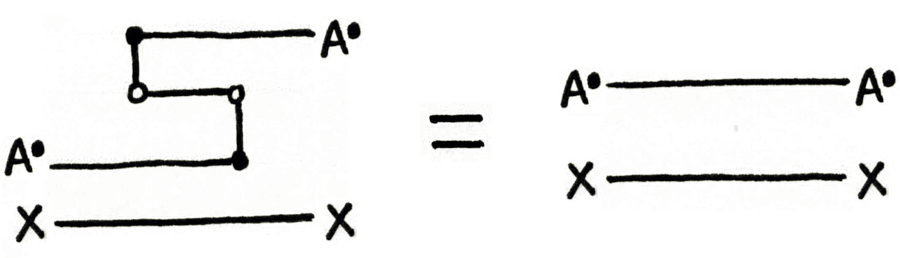}
\end{mathpar}
We therefore conclude that $A \circ - \dashv A \bullet -$, as required.
\end{proof}

Now, every monoidal category is a linearly distributive category (with both monoidal operations given by $\otimes$), and it turns out that $\bh\,\corner{\A}$ forms a (somewhat degenerate) linear actegory. Of the nine natural families of arrows required by the definition, four are accounted for by the isomorphisms of Lemma~\ref{lem:moraliso}, a further four become identities in our setting, and the final one is given by the $d^\circ_\bullet$ morphisms from Observation~\ref{obs:halfsym}. The coherence conditions all hold trivially. We record:
\begin{prop}\label{prop:linearactegory}
  Let $\A$ be a resource theory. Then $\bh\,\corner{\A}$ is a linear actegory.
\end{prop}

This is intriguing insofar as it exhibits a formal connection between the free cornering of a resource theory and existing work on behavioural types. For example, the message-passing interpretation of classical linear logic presented by Wadler in~\cite{Wad14} corresponds to the message-passing interpretation of linear actegories in the special case of a *-autonomous category acting on itself (Example 4.2(4) of~\cite{Coc09}). There may be an even stronger connection to the behavioural type interpretation of intuitionistic linear logic due to Caires and Pfenning~\cite{Caires2010}, although here the connection to the logic of message passing is weaker (Example 4.2(1) of~\cite{Coc09}). We leave the full investigation of these connections for future work.

\section{Axioms for Resource Transducers}\label{sec:axiomscheme}

We have seen that the category of horizontal cells of the free cornering of a resource theory is an interesting object of study in its own right: it is a planar monoidal category that arises naturally and is different from those typically considered. In this section we give a direct presentation of $\bh\,\corner{\A}$ both to deepen our understanding of its structure and to facilitate its use as an example (or counterexample) in the future. While there are many axioms, they are mostly intuitive, and are conveniently organized into pairs by the contravariant involution $(-)^*$ of Lemma~\ref{lem:involution}.

Let $\A$ be a resource theory. Define $\mathsf{T}(\A)$ to be the free spatial strict monoidal category with the generating objects as in:
\begin{mathpar}
  \inferrule{A \in \A_0}{A^\circ \quad \mathsf{obj}}

  \inferrule{A \in \A_0}{A^\bullet \quad \mathsf{obj}}
\end{mathpar}
and the generating morphisms given by:
\begin{mathpar}
  \inferrule*[Right=$\circ$]
             {f : A \to B \in \A_1}
             {f^\circ : A^\circ \to B^\circ}

  \inferrule*[Right=$\bullet$]
             {f : A \to B \in \A_1}
             {f^\bullet : B^\bullet \to A^\bullet}
\\

  \inferrule*[Right=$\triangleleft$]
            {A,B \in \A_0}
            {\triangleleft_{A,B} : (A \otimes B)^\circ \to A^\circ \otimes B^\circ}

  \inferrule*[Right=$\blacktriangleright$]
             {A,B \in \A_0}
             {\blacktriangleright_{A,B}\, : A^\bullet \otimes B^\bullet \to (A \otimes B)^\bullet }
\\
  \inferrule*[Right=$\triangleright$]
            {A,B \in \A_0}
            {\triangleright_{A,B} : A^\circ \otimes B^\circ \to (A \otimes B)^\circ}

  \inferrule*[Right=$\blacktriangleleft$]
             {A,B \in \A_0}
             {\blacktriangleleft_{A,B}\, : (A \otimes B)^\bullet \to A^\bullet \otimes B^\bullet}
            \\
  \inferrule*[Right=$\multimapinv$]
             {\text{}}
             {\multimapinv\, : I \to I^\circ}

  \inferrule*[Right=$\multimap$]
             {\text{}}
             {\multimap\, : I^\circ \to I}

  \inferrule*[Right=$\multimapdot$]
             {\text{}}
             {\multimapdot\, : I^\bullet \to I}

  \inferrule*[Right=$\multimapdotinv$]
             {\text{}}
             {\multimapdotinv\, : I \to I^\bullet}
\\
  \inferrule*[Right=$\sigma\circ$]
             {A,B \in \A_0}
             {\sigma^\circ_{A,B} : A^\circ \otimes B^\circ \to B^\circ \otimes A^\circ}

  \inferrule*[Right=$\sigma\bullet$]
             {A,B \in \A_0}
             {\sigma^\bullet_{A,B} : A^\bullet \otimes B^\bullet \to B^\bullet \otimes A^\bullet}
\\
  \inferrule*[Right=$\eta$]
             {A \in \A_0}
             {\eta_A : I \to A^\bullet \otimes A^\circ}

  \inferrule*[Right=$\varepsilon$]
             {A \in \A_0}
             {\varepsilon_A : A^\circ \otimes A^\bullet \to I}
\end{mathpar}
The rules $\circ$ and $\bullet$ correspond to the image of the functors from Lemma~\ref{lem:strongmonoidal}. All of $\triangleleft,\triangleright,\blacktriangleleft,\blacktriangleright,\multimap,\multimapinv,\multimapdot,\multimapdotinv,\sigma\circ, \sigma\bullet$ correspond to the isomorphisms of Lemma~\ref{lem:moraliso}, and the $\eta$ and $\varepsilon$ rules correspond to the morphisms considered at the end of Section~\ref{sec:transducers} that lead to Proposition~\ref{prop:linearactegory}.

Before presenting the equations for $\mathsf{T}(\A)$ we give the following string-diagrammatic conventions for our generators:
\begin{mathparpagebreakable}
    f^\circ \hspace{0.2cm}\leftrightsquigarrow\hspace{0.2cm} \includegraphics[height=1cm,align=c]{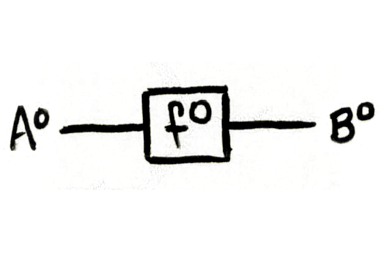}

  f^\bullet \hspace{0.2cm}\leftrightsquigarrow\hspace{0.2cm} \includegraphics[height=1cm,align=c]{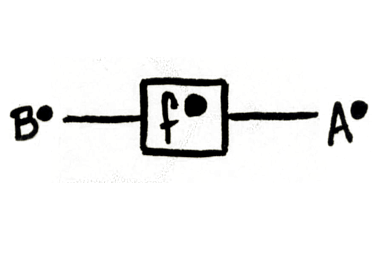}
\\
  \triangleleft_{A,B} \hspace{0.2cm}\leftrightsquigarrow\hspace{0.2cm} \includegraphics[height=1cm,align=c]{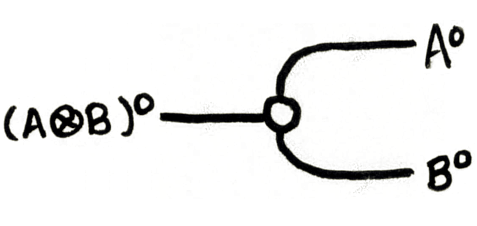}

  \blacktriangleright_{A,B}\, \hspace{0.2cm}\leftrightsquigarrow\hspace{0.2cm} \includegraphics[height=1cm,align=c]{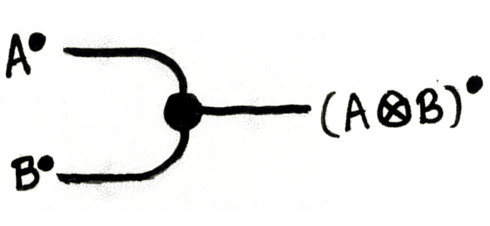}
\\
  \triangleright_{A,B} \hspace{0.2cm}\leftrightsquigarrow\hspace{0.2cm} \includegraphics[height=1cm,align=c]{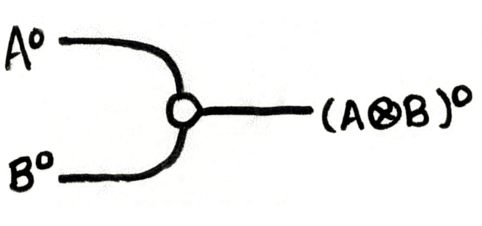}

  \blacktriangleleft_{A,B}\, \hspace{0.2cm}\leftrightsquigarrow\hspace{0.2cm} \includegraphics[height=1cm,align=c]{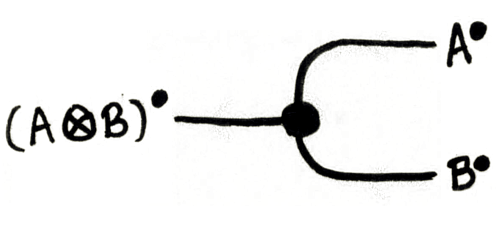}
\\
  \multimapinv\, \hspace{0.2cm}\leftrightsquigarrow\hspace{0.2cm} \includegraphics[height=1cm,align=c]{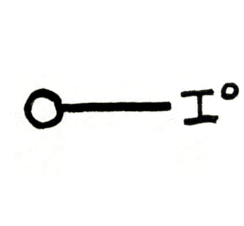}

  \multimap\, \hspace{0.2cm}\leftrightsquigarrow\hspace{0.2cm} \includegraphics[height=1cm,align=c]{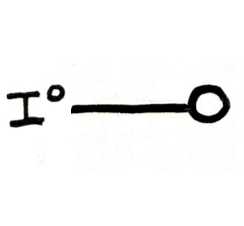}

  \multimapdot\, \hspace{0.2cm}\leftrightsquigarrow\hspace{0.2cm} \includegraphics[height=1cm,align=c]{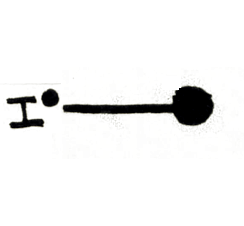}

  \multimapdotinv\, \hspace{0.2cm}\leftrightsquigarrow\hspace{0.2cm} \includegraphics[height=1cm,align=c]{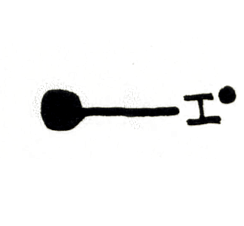}
\\
  \sigma_{A,B}^\circ \hspace{0.2cm}\leftrightsquigarrow\hspace{0.2cm} \includegraphics[height=1cm,align=c]{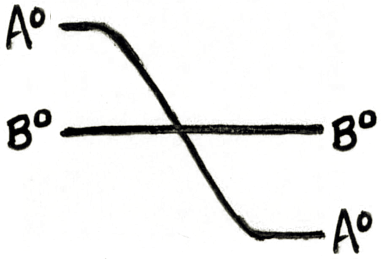}

  \sigma_{A,B}^\bullet \hspace{0.2cm}\leftrightsquigarrow\hspace{0.2cm} \includegraphics[height=1cm,align=c]{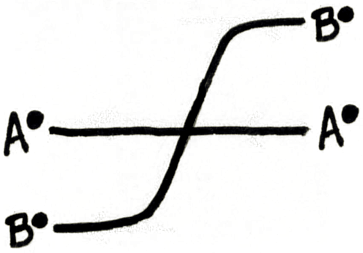}
\\

  \eta_A \hspace{0.2cm}\leftrightsquigarrow\hspace{0.2cm} \includegraphics[height=1cm,align=c]{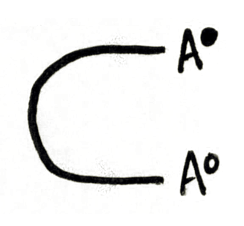}

  \varepsilon_A \hspace{0.2cm}\leftrightsquigarrow\hspace{0.2cm} \includegraphics[height=1cm,align=c]{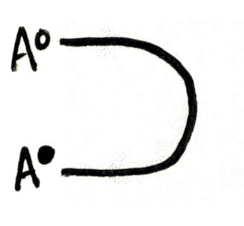}
\end{mathparpagebreakable}
While it is initially difficult to keep track of the polarity of each wire, particularly in the diagrams for the $\sigma\circ$, $\sigma\bullet$, $\eta$, and $\varepsilon$ morphisms, this is alleviated by the fact that resources may flow down but not up. Keeping this in mind allows us to omit any sort of directional information from the wires of our diagrams, which we feel makes them more readable.

Now, we impose the following equations in addition to those of a spatial strict monoidal category, and those inherited from $\A$. First, concerning the interaction of the $\triangleleft,\triangleright$ and $\blacktriangleleft,\blacktriangleright$ morphisms we require:
\begin{mathpar}
  \includegraphics[height=1cm,align=c]{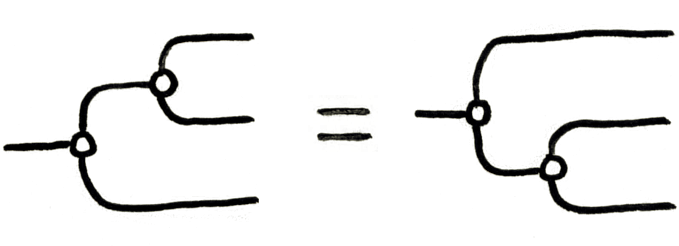}

  \includegraphics[height=1cm,align=c]{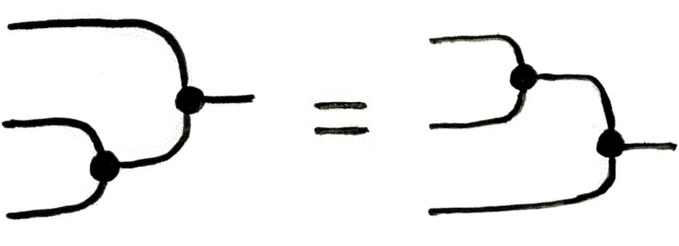}

  \includegraphics[height=1cm,align=c]{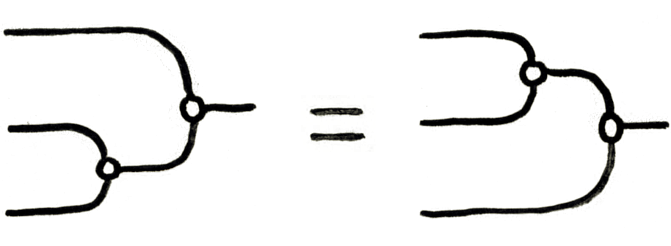}

  \includegraphics[height=1cm,align=c]{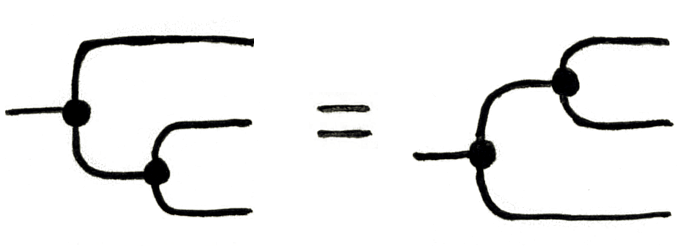}

  \includegraphics[height=1cm,align=c]{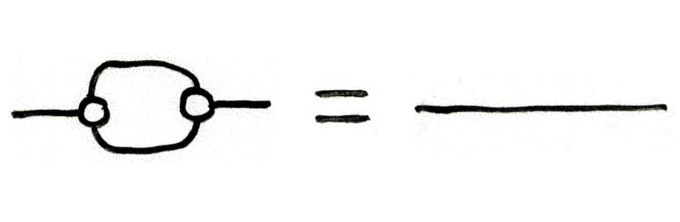}

  \includegraphics[height=1cm,align=c]{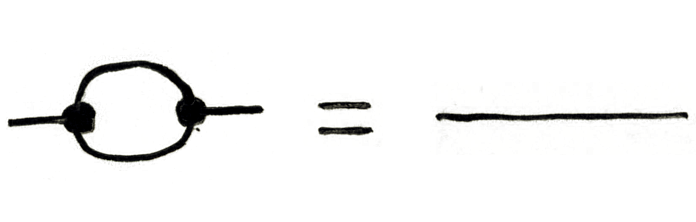}

  \includegraphics[height=1cm,align=c]{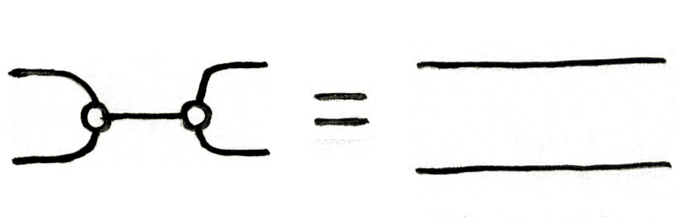}

  \includegraphics[height=1cm,align=c]{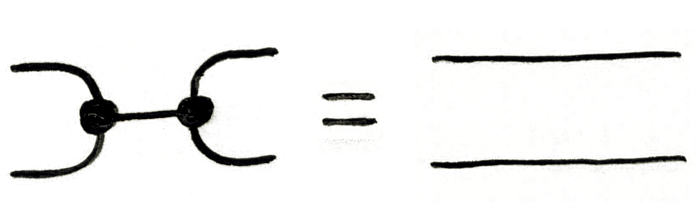}
\end{mathpar}
We note that this is a polarized version of the axioms for ``dividers'' and ``gatherers'' found in the SZX calculus~\cite{Car2019}. We continue with axioms concerning the interaction of the $\multimap,\multimapinv$ and $\multimapdot,\multimapdotinv$ morphisms:
\begin{mathpar}
  \includegraphics[height=1cm,align=c]{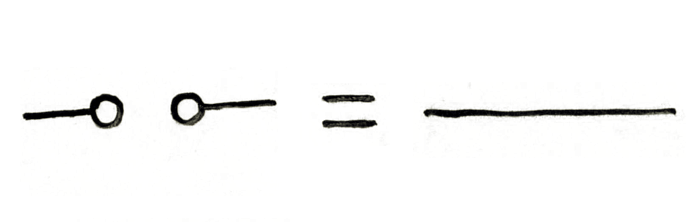}

  \includegraphics[height=1cm,align=c]{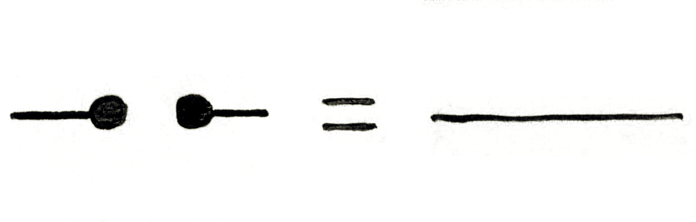}

  \includegraphics[height=1cm,align=c]{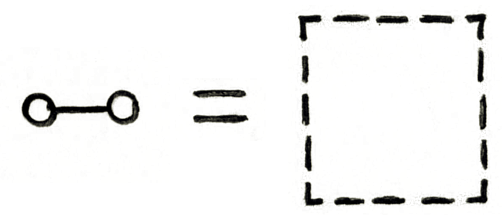}

  \includegraphics[height=1cm,align=c]{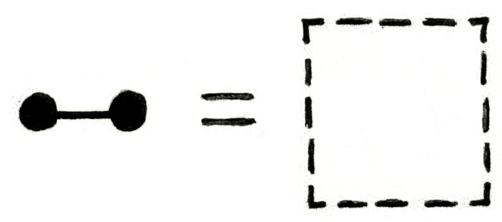}
\end{mathpar}
For the remaining interactions of $\triangleleft,\triangleright,\multimap,\multimapinv$ and $\blacktriangleleft,\blacktriangleright,\multimapdot,\multimapdotinv$ we require:
\begin{mathpar}
  \includegraphics[height=1cm,align=c]{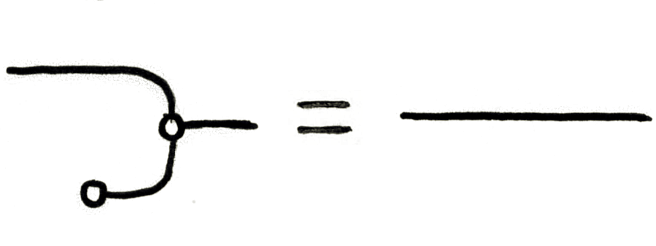}

  \includegraphics[height=1cm,align=c]{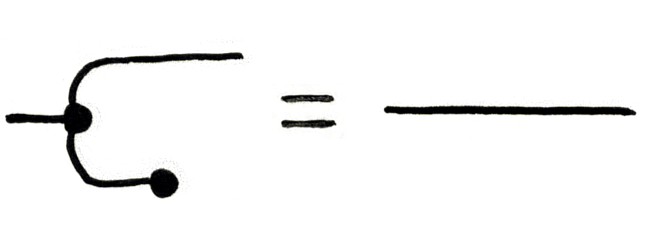}

  \includegraphics[height=1cm,align=c]{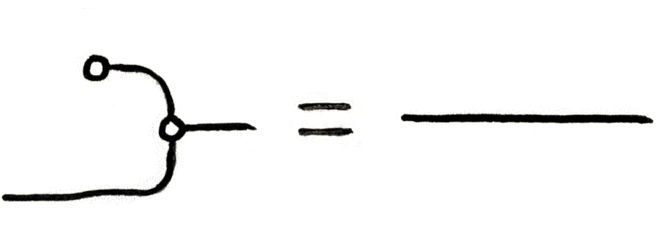}

  \includegraphics[height=1cm,align=c]{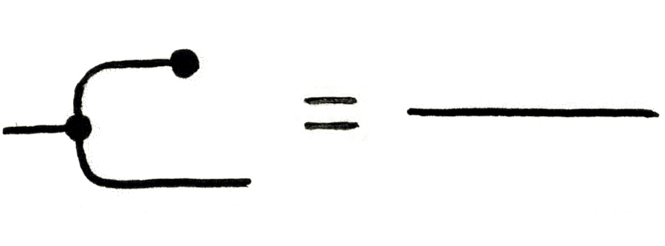}

  \includegraphics[height=1cm,align=c]{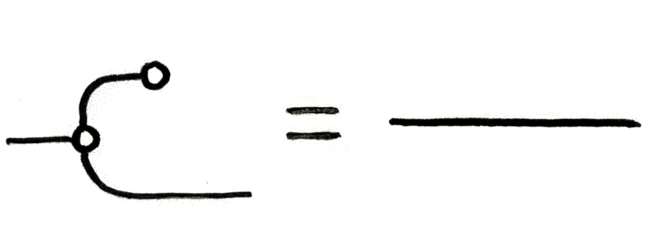}

  \includegraphics[height=1cm,align=c]{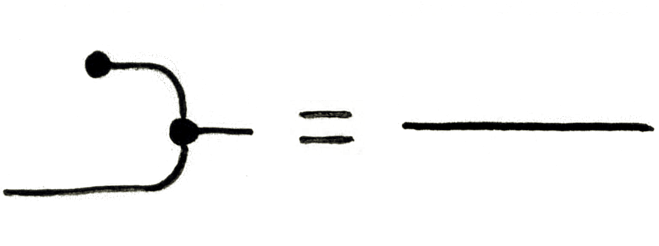}

  \includegraphics[height=1cm,align=c]{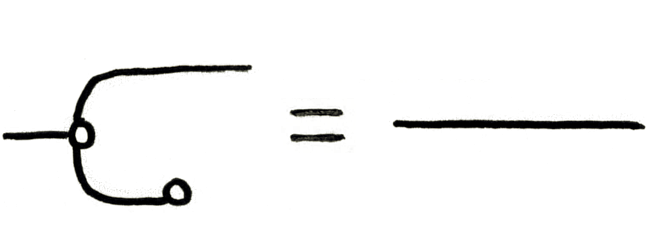}

  \includegraphics[height=1cm,align=c]{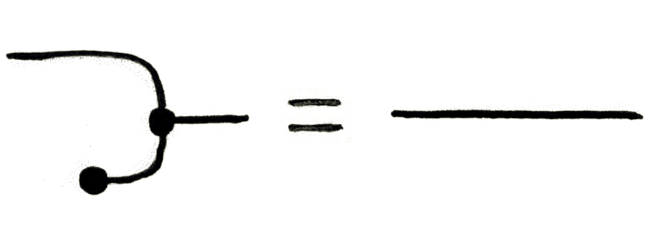}

  \includegraphics[height=1cm,align=c]{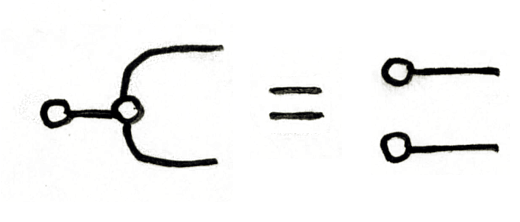}

  \includegraphics[height=1cm,align=c]{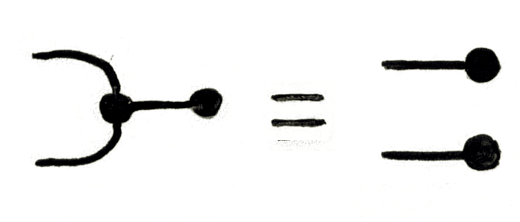}

  \includegraphics[height=1cm,align=c]{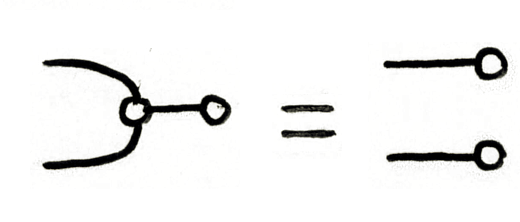}

  \includegraphics[height=1cm,align=c]{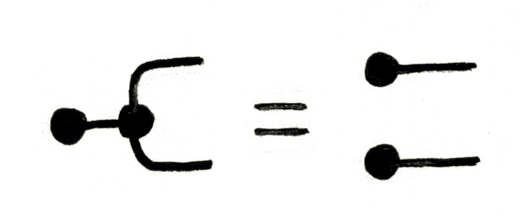}
\end{mathpar}
Next, the interaction between $\sigma\circ,\circ$ and $\sigma\bullet,\bullet$ is captured by:
\begin{mathpar}
  \includegraphics[height=1cm,align=c]{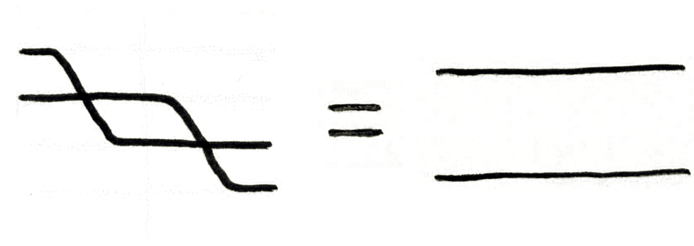}

  \includegraphics[height=1cm,align=c]{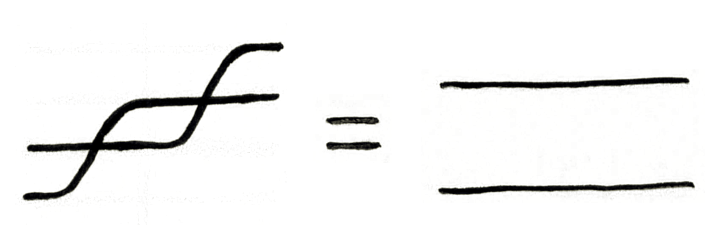}

  \includegraphics[height=1cm,align=c]{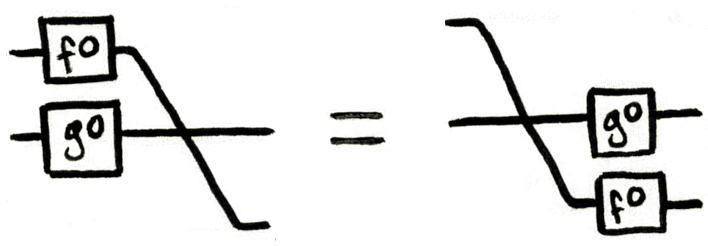}

  \includegraphics[height=1cm,align=c]{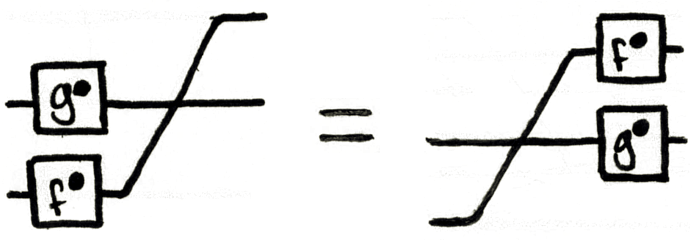}
\end{mathpar}
For the interaction between $\sigma\circ,\triangleleft,\triangleright$ and $\sigma\bullet,\blacktriangleleft,\blacktriangleright$ we require:
\begin{mathpar}
  \includegraphics[height=1cm,align=c]{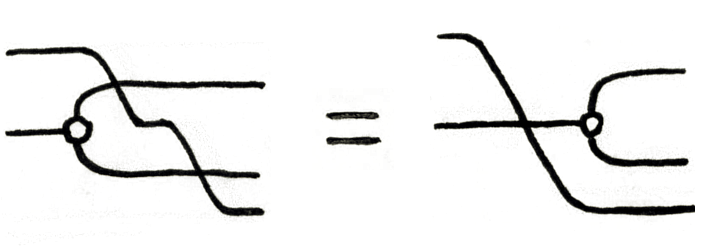}

  \includegraphics[height=1cm,align=c]{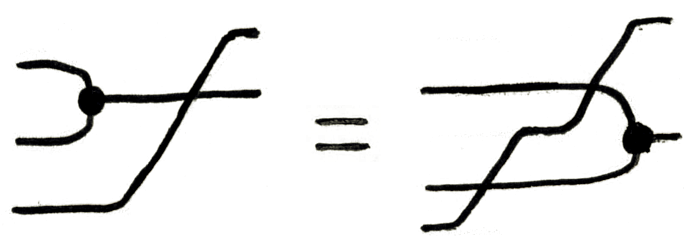}

  \includegraphics[height=1cm,align=c]{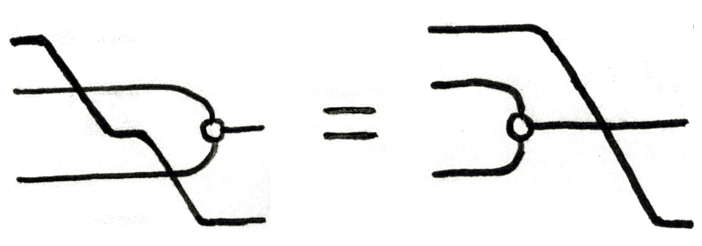}

  \includegraphics[height=1cm,align=c]{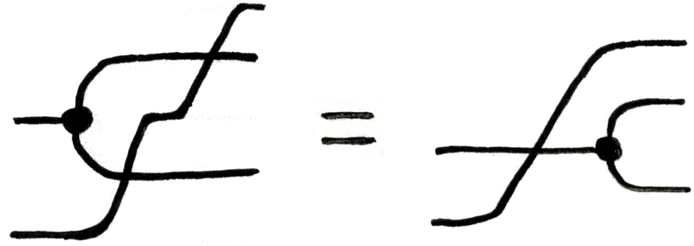}

  \includegraphics[height=1cm,align=c]{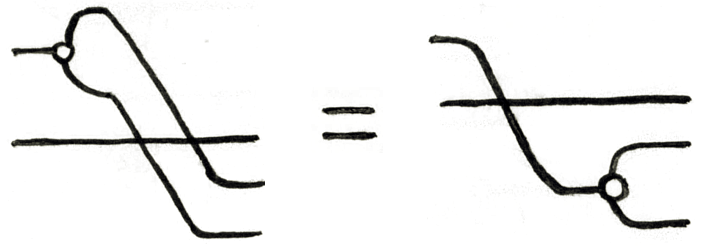}

  \includegraphics[height=1cm,align=c]{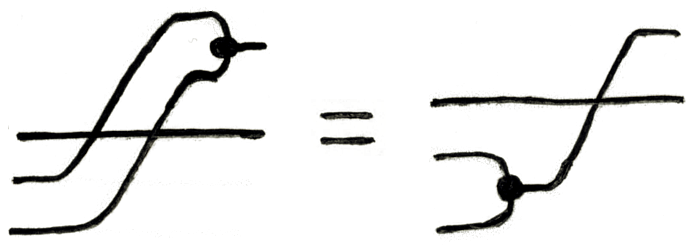}

  \includegraphics[height=1cm,align=c]{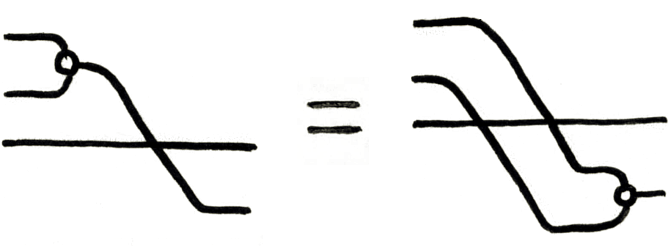}

  \includegraphics[height=1cm,align=c]{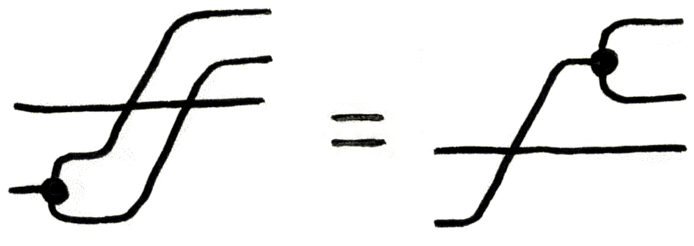}
\end{mathpar}
and the interaction between $\sigma\circ,\multimap,\multimapinv$ and $\sigma\bullet,\multimapdot,\multimapdotinv$ is captured by:
\begin{mathpar}
  \includegraphics[height=1cm,align=c]{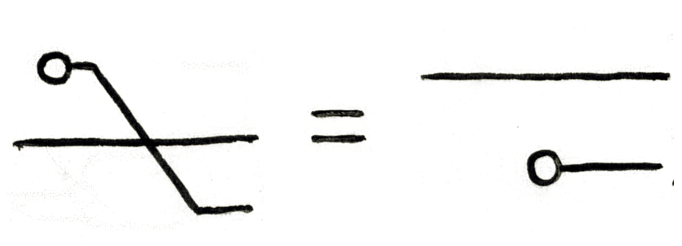}

  \includegraphics[height=1cm,align=c]{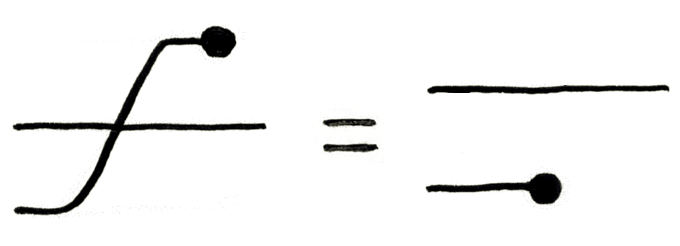}

  \includegraphics[height=1cm,align=c]{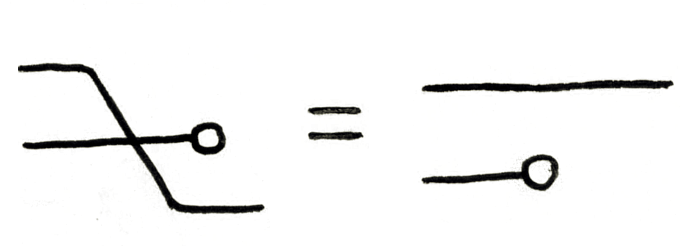}

  \includegraphics[height=1cm,align=c]{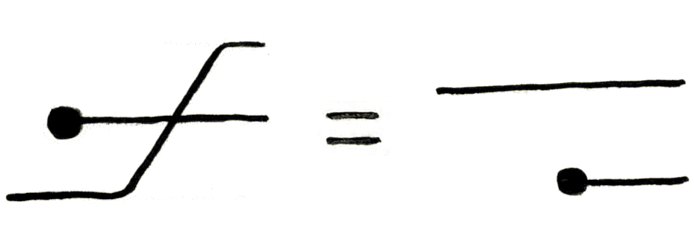}

  \includegraphics[height=1cm,align=c]{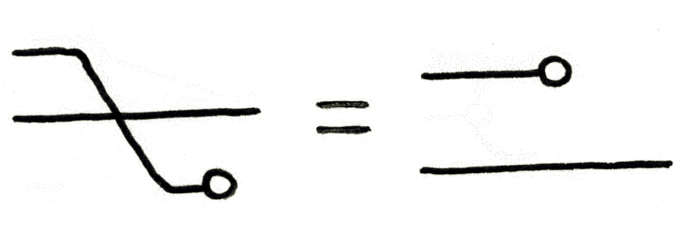}

  \includegraphics[height=1cm,align=c]{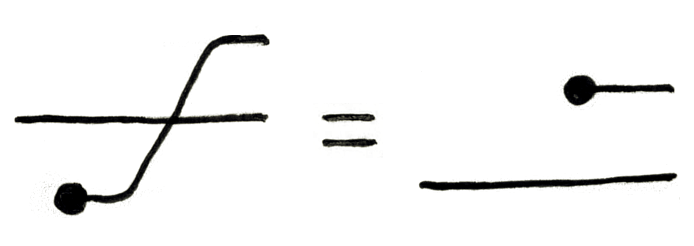}

  \includegraphics[height=1cm,align=c]{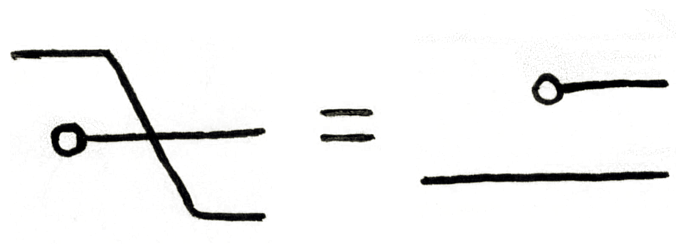}

  \includegraphics[height=1cm,align=c]{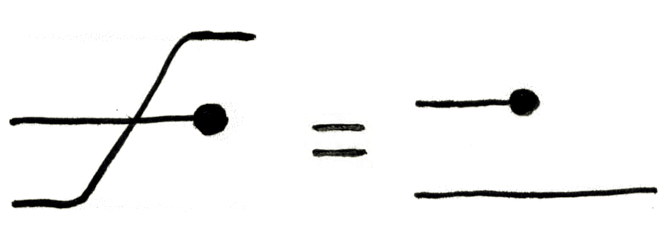}
\end{mathpar}
For the interaction between $\eta,\varepsilon,\circ,\bullet$ we require:
\begin{mathpar}
  \includegraphics[height=1cm,align=c]{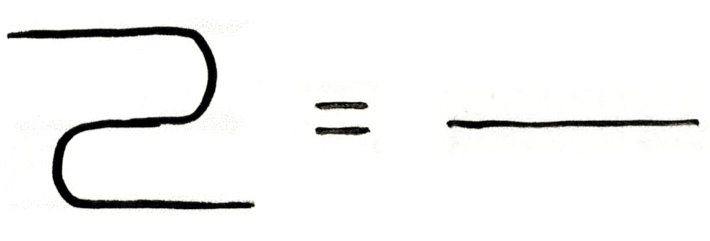}

  \includegraphics[height=1cm,align=c]{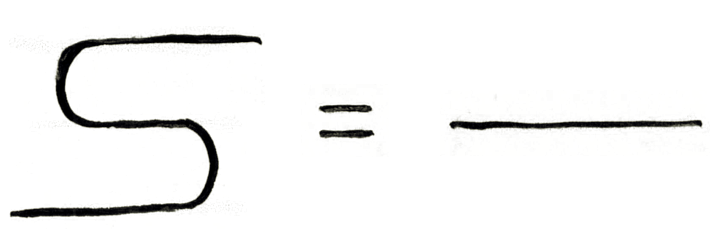}

  \includegraphics[height=1cm,align=c]{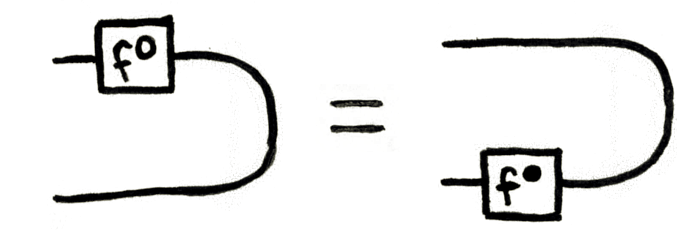}

  \includegraphics[height=1cm,align=c]{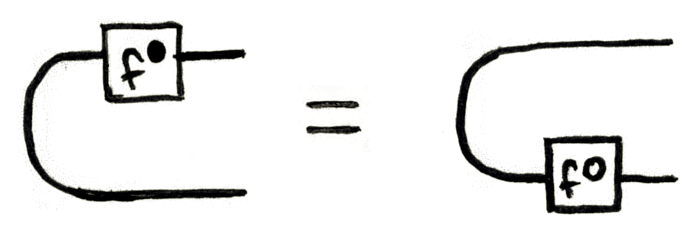}
\end{mathpar}
The interaction between $\sigma\circ,\sigma\bullet,\eta,\varepsilon$ is captured by:
\begin{mathpar}
  \includegraphics[height=1cm,align=c]{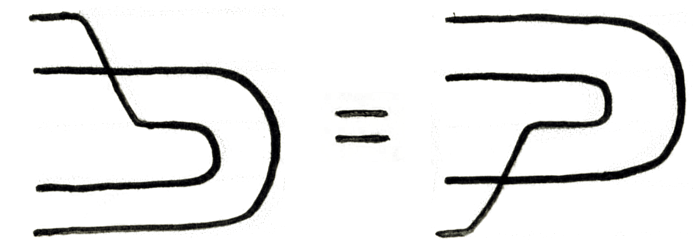}

  \includegraphics[height=1cm,align=c]{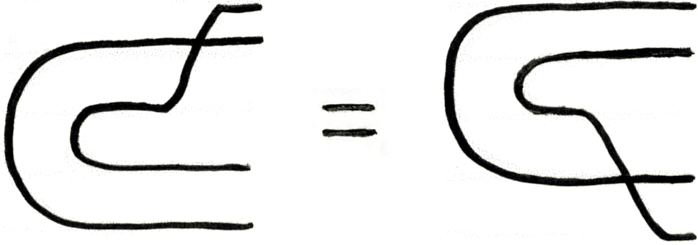}
\end{mathpar}
And for the interaction between $\eta,\varepsilon,\triangleleft,\triangleright,\blacktriangleleft,\blacktriangleright$ we ask that:
\begin{mathpar}
  \includegraphics[height=1cm,align=c]{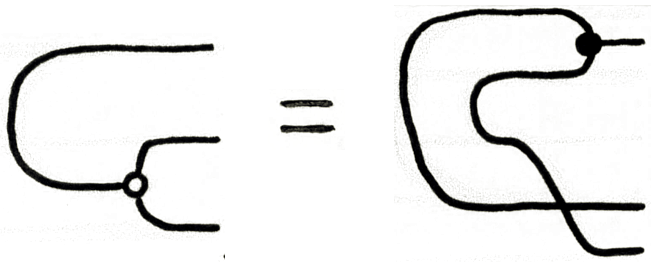}

  \includegraphics[height=1cm,align=c]{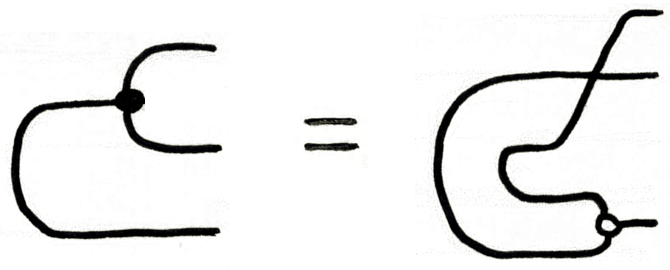}

  \includegraphics[height=1cm,align=c]{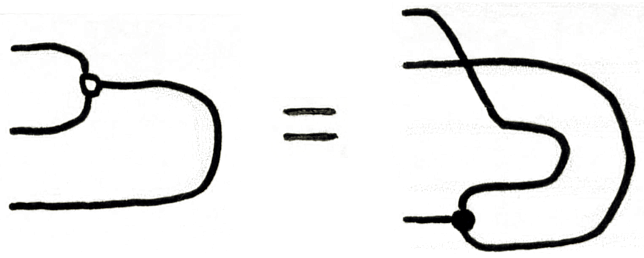}

  \includegraphics[height=1cm,align=c]{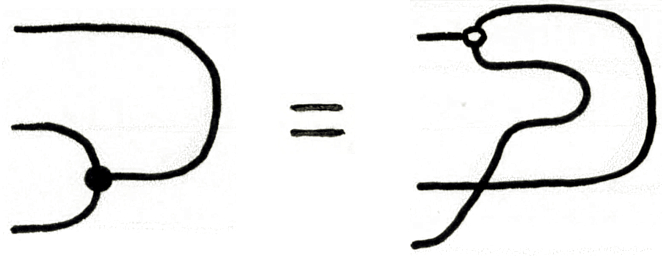}
\end{mathpar}
Finally, we require the following axioms concerning $f^\circ$ and $f^\bullet$:
\begin{mathpar}

  f^\circ g^\circ = (fg)^\circ

  \triangleleft (f^\circ \otimes g^\circ) \,\triangleright = (f \otimes g)^\circ

  (1_A)^\circ = 1_{A^\circ}

  \triangleright (\sigma_{A,B})^\circ \triangleleft = \sigma^\circ_{A,B}

  \\ %white ends black begins

  g^\bullet f^\bullet = (fg)^\bullet

  \blacktriangleleft \!(f^\bullet \otimes g^\bullet) \!\blacktriangleright\, = (f \otimes g)^\bullet

  (1_A)^\bullet = 1_{A^\bullet}

  \blacktriangleright \!(\sigma_{A,B})^\bullet \!\!\blacktriangleleft\, = \sigma^\bullet_{A,B}
\end{mathpar}
This concludes the presentation of $\mathsf{T}(\A)$. We proceed to define a strict monoidal functor $M : \mathsf{T}(\A) \to \bh\,\corner{\A}$ on objects by $M(X) = X$ (since $\mathsf{T}(\A)$ and $\bh\,\corner{\A}$ have the same objects) and on the generators by:
\begin{mathpar}
{
  M(f^\circ) =
  \includegraphics[height=1cm,align=c]{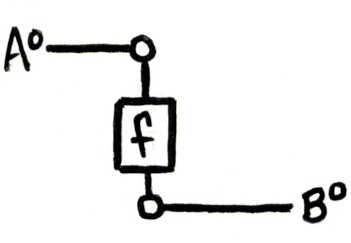}
}

{
  M(f^\bullet) =
  \includegraphics[height=1cm,align=c]{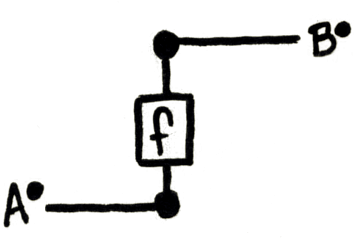}
}
\\
{
  M(\triangleleft_{A,B}) =
  \includegraphics[height=1cm,align=c]{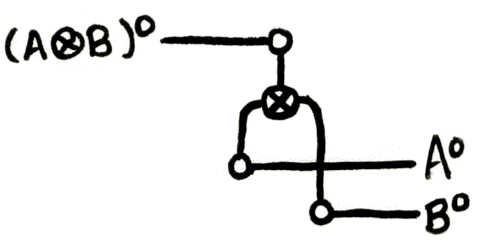}
}

{
  M(\blacktriangleright_{A,B}) =
  \includegraphics[height=1cm,align=c]{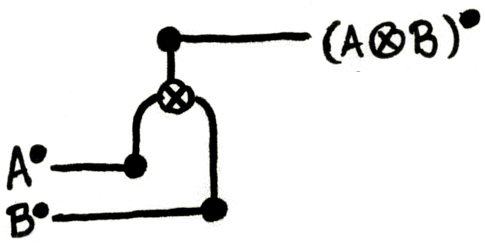}
}
\\
{
  M(\triangleright_{A,B}) =
  \includegraphics[height=1cm,align=c]{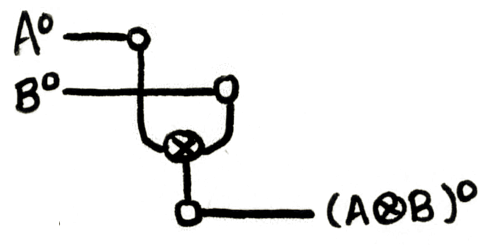}
}

{
  M(\blacktriangleleft_{A,B}) =
  \includegraphics[height=1cm,align=c]{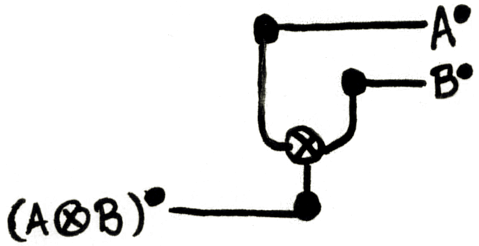}
}
\\
{
  M(\multimapinv) =
  \includegraphics[height=1cm,align=c]{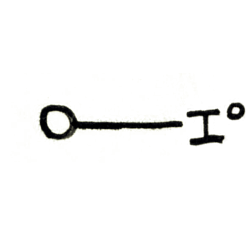}
}

{
  M(\multimap) =
  \includegraphics[height=1cm,align=c]{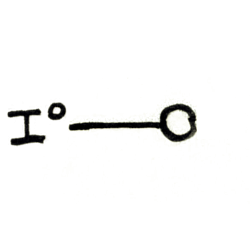}
}

{
  M(\multimapdot) =
  \includegraphics[height=1cm,align=c]{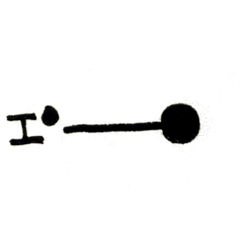}
}

{
  M(\multimapdotinv) =
  \includegraphics[height=1cm,align=c]{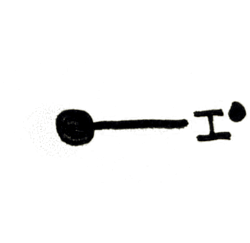}
}
\\
{
  M(\sigma_{A,B}^\circ) =
  \includegraphics[height=1cm,align=c]{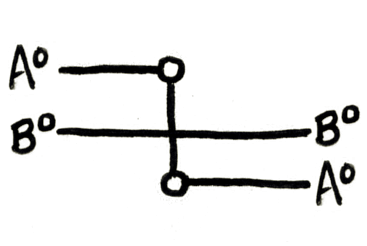}
}

{
  M(\sigma_{A,B}^\bullet) =
  \includegraphics[height=1cm,align=c]{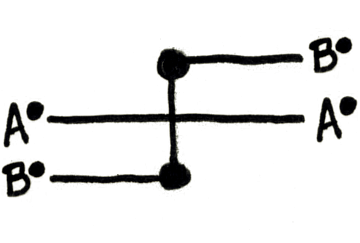}
}
\\
{
  M(\eta_A) =
  \includegraphics[height=1cm,align=c]{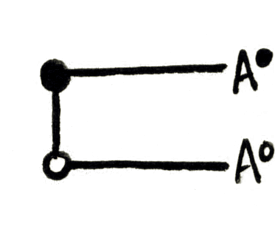}
}

{
  M(\varepsilon_A) =
  \includegraphics[height=1cm,align=c]{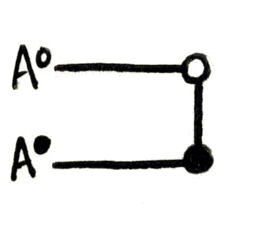}
}
\end{mathpar}
It is straightforward to verify that $M$ is a strict monoidal functor. Additionally, we have:
\begin{prop}
  $M : \mathsf{T}(\A) \to \bh\,\corner{\A}$ is full, faithful, and identity-on-objects.
\end{prop}
\begin{proof}
  That $M$ is full follows from the cohrerence theorem for string diagrams for proarrow equipments~\cite{Mye16}. Intuitively, every arrow of $\bh\,\corner{\A}$ is either in the image of $(-)^\circ$ or $(-)^\bullet$, or is built out of corner cells and crossing cells. Every horizontal cell of $\bh\,\corner{\A}$ that can be built out of only corner cells and does not decompose into multiple such cells is the image of one of the generators of $\mathsf{T}(\A)$, and so we know that $M$ is full. Perhaps surprising is that the horizontal cell $d^\circ_\bullet : A^\circ \otimes B^\bullet \to B^\bullet \otimes A^\circ$ of Observation~\ref{obs:halfsym} decomposes in this way, being the image under $M$ of the following morphism in $\mathsf{T}(\A)$:
\[
  \includegraphics[height=1.2cm,align=c]{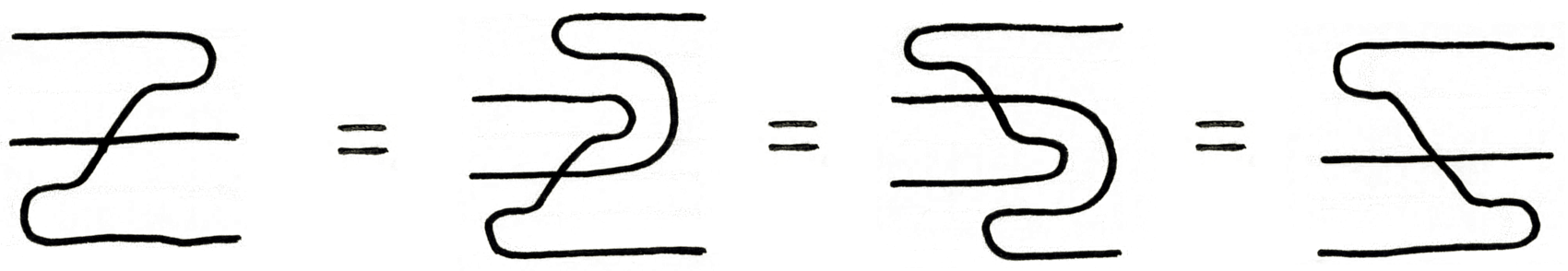}
\]
To show that $M$ is faithful is to show that the equations of $\mathsf{T}(\A)$ capture all equations between horizontal cells of $\corner{\A}$ when taken together with the equations of a spatial strict monoidal category. Recall that all of the equations of $\corner{\A}$ are generated by the yanking equations, along with any equations of $\A$. The yanking equations are local, in that each instance of one of the yanking equations involves exactly two cells of $\corner{\A}$, so we need only consider local interactions of cells of $\bh\,\corner{\A}$ in our analysis. It is relatively straightforward to verify that the defining equations of $\mathsf{T}(\A)$ are precisely the equations that arise in this way, and so $M$ is faithful.\footnote{Given the large number of equations involved, it is of course possible that we have missed some. That said, we are reasonably confident that the equations we have given are complete in this sense.} Finally, $M$ is clearly identity-on-objects.
\end{proof}
It follows that our axiomatization of $\bh\,\corner{\A}$ is correct. We record:
\begin{cor}\label{cor:axiomswork}
There is an isomorphism of categories $\bh\,\corner{\A} \cong \mathsf{T}(\A)$.
\end{cor}

\section{Conclusions and Future Work}\label{sec:conclusion}

We have shown how to decompose the material history of a process into concurrent components by working in the free cornering of an appropriate resource theory. We have explored the structure of the free cornering in light of this interpretation and found that it is consistent with our intuition about how this sort of thing ought to work. We do not claim to have solved all problems in the modelling of concurrency, but we feel that our formalism captures the material aspect of concurrent systems very  well.

We find it quite surprising that the structure required to model concurrent resource transformations is precisely the structure of a proarrow equipment. This structure is already known to be important in formal category theory, and we are appropriately intrigued by its apparent relevance to models of concurrency --- a far more concrete setting than the usual context in which one encounters proarrow equipments!

Further, we have considered categories of resource transducers that are induced by our construction. We have identified some structure they do and do not exhibit, and have provided a more direct axiomatization of them. We are not aware of any categories with similar structure, which we feel makes these categories of resource transducers worthy of further study, and of potential value as a counterexample.

There are of course many directions for future work. For one, it would be nice to connect the development here to the wider literature on concurrent processes. An obstacle to this is that the free cornering does not allow us to express branching or recursion, both of which feature heavily in more general theories of process communication. If we assume that our monoidal category $\A$ has binary coproducts then we may represent a limited sort of branching computation in which $(A + B)^\circ$ and $(A + B)^\bullet$ represent choices to be made by the left and right participant respectively, but this is less flexible than the protocol-level choice that one finds in e.g.\ session types or the nondeterminism of process calculi. We speculate that this is best approached through the ``situated transition systems'' introduced in~\cite{Nester2021Situated}, in which the concurrent resource transformations developed in~\cite{Nester2021Concurrent} (which this paper extends) are used to augment the category of spans of reflexive graphs --- interpreted as open transition systems~\cite{Kat97} --- to generate material history over some resource theory as transitions unfold in time. Alternatively, one might impose additional structure on the free cornering to allow nondeterministic choice and repetition.

Another direction for future work is to pursue the connection with the message passing logic of Cockett and Pastro~\cite{Coc09} (established in Proposition~\ref{prop:linearactegory}) and the wider programme of behavioural types influenced by linear logic including~\cite{Wad14} and~\cite{Caires2010}. Finally, the presence of proarrow equipments here is rather mysterious, and we wonder if some deeper reason for it might exist.

\bibliographystyle{alphaurl}
\bibliography{citations}

\end{document}